\documentclass[11pt,a4paper]{article}
\usepackage{subcaption}

\usepackage{amsmath,amssymb,amsthm,bm,mathrsfs}
\usepackage{xcolor}
\usepackage{booktabs}
\usepackage{natbib}
\usepackage{authblk}
\usepackage{threeparttable}
\usepackage{multirow}
\usepackage[colorlinks=true,linkcolor=blue,citecolor=blue,urlcolor=blue]{hyperref}
\usepackage{graphicx}
\usepackage{makecell}
\newtheorem{corollary}{Corollary}
\newtheorem{theorem}{Theorem}
\newtheorem{lemma}{Lemma}
\newtheorem{assumption}{Assumption}

\newtheorem{remark}{Remark}

\newcommand{\Prob}{\mathbb{P}}

\title{\textbf{Testing Conditional Stochastic Dominance via Copula Derivatives}}
		
        \author[a]{Weiqi Yang}
        \author[a]{Weiwei Zhuang}
		\author[b]{Xiaojun Song\thanks{Corresponding author: \texttt{sxj@gsm.pku.edu.cn} (X. Song).}}
		
		\affil[a]{School of Management, University of Science and Technology of China, China}
		\affil[b]{Guanghua School of Management, Peking University, China}
		
\date{\today}

\begin{document}
\maketitle
\begin{abstract}
Comparing two populations at the same physical covariate value requires more than conditional means or isolated target-point decisions: researchers may need evidence about an entire conditional-distribution ordering over a continuum, even when covariate margins differ. This paper makes that common-value comparison estimable under an explicit structure--flexibility tradeoff and turns the resulting surface into simultaneous evidence for first-order stochastic dominance. Population-specific margins map the common covariate value into each group, while a fitted copula-derivative representation links conditional distributions across the region. Uniform inference propagates uncertainty from both the margins and dependence model through a one-sided statistic with unknown binding locations. Under correct specification within a finite copula class, smoothness and trimming conditions, and a uniquely best candidate family, the procedure admits uniform control and consistent calibration. Simulations show increasing rejection as alternatives become more distinguishable, alongside model-selection sensitivity and small-sample size distortion. In a descriptive PSID application, the high--low parental-education comparison satisfies the two-direction criterion after multiplicity adjustment, whereas adjacent education-group comparisons remain inconclusive. The framework therefore supports region-wide distributional comparison while making its structural and inferential boundaries explicit.
\end{abstract}

\section{Introduction}

Comparing populations at a common covariate value is often a distributional problem rather than a question about conditional means. Income opportunities, treatment responses, risks, and performance can differ in dispersion, tails, or crossing patterns even when their averages are similar. The relevant object is therefore the full conditional distribution, $H_g(y\mid x)=\Prob(Y_g\leq y\mid X_g=x)$. Researchers may also need to know whether an ordering holds throughout an empirically relevant conditioning region, not only at selected values. This distinction is especially consequential when populations have different covariate margins: a common percentile and a common physical value of $x$ then define different comparisons. This paper studies first-order dominance at common values of a continuous covariate over a trimmed continuum of conditioning and outcome values, while allowing each population to retain its own covariate margin.

The paper belongs to a broad inferential lineage rather than a narrow estimator class. Unconditional stochastic-dominance testing established how distributional inequalities can be evaluated without reducing welfare or risk comparisons to a few moments \citep{McFadden1989,Anderson1996,DavidsonDuclos2000,BarrettDonald2003,LintonMaasoumiWhang2005,LintonSongWhang2010}. Conditional work has since developed several complementary branches: global tests and conditional-moment formulations \citep{DelgadoEscanciano2013,AndrewsShi2013,AndrewsShi2017}, conditional treatment and distributional-effect tests \citep{Abadie2002,ChangLeeWhang2015}, inference at local or target covariate values \citep{DonaldHsuBarrett2012,ShenZhang2016,GoldmanKaplan2018,QuYoon2019,BugniCanayKim2025}, and structured or dynamic dominance procedures \citep{GonzaloOlmo2014,LintonSeoWhang2023}. These branches establish a substantial intellectual neighborhood for comparisons of conditional distributions, but they address different conditioning domains, statistical objects, and inferential targets.

The principal benchmarks sharpen the research question at hand. Delgado and Escanciano test global conditional-dominance restrictions through integrated joint-distribution differences \citep{DelgadoEscanciano2013}, while Chang, Lee, and Whang directly test conditional stochastic dominance between treatment and control distributions over covariate values \citep{ChangLeeWhang2015}. Linton, Seo, and Whang study first- and second-order dominance with high- or growing-dimensional dynamic information under a location–scale structure \citep{LintonSeoWhang2023}, and Gonzalo and Olmo develop conditional-dominance tests in dynamic settings \citep{GonzaloOlmo2014}. These papers establish that neither continuum-wide CSD nor structured conditional dominance is new. A distinct finite/local branch asks for distributional inference at selected conditioning values: Bugni, Canay, and Kim consider one or finitely many prespecified targets \citep{BugniCanayKim2025}, with related cutoff and local-distribution procedures developed by Shen and Zhang and Goldman and Kaplan \citep{ShenZhang2016,GoldmanKaplan2018}. The present paper addresses a different inferential target: one simultaneous two-population ordering over a trimmed continuum of common physical $(x,y)$ values under separate margins and an explicit dependence class. This difference matters because local decisions cannot authorize a region-wide conclusion, while the principal global benchmarks use different statistical representations and maintained structures.

Moving from selected targets to a surface changes the inferential problem. The inequality may bind at unknown locations on a two-dimensional continuum, while estimated marginal mappings and dependence parameters affect every point of that surface. The paper addresses this difficulty through a copula derivative. The established representation links a conditional distribution to the derivative of its joint copula \citep{JanssenSwanepoelVeraverbeke2017}: population-specific margins locate the same physical $x$ within each population, while the copula supplies coherence across conditioning values. Combined with semiparametric copula estimation \citep{GenestGhoudiRivest1995,Tsukahara2005}, this yields one fitted conditional-distribution surface for each population. The copula derivative is thus enabling literature, not the novelty protagonist. It is useful here because its separation of marginal location and dependence matches the added common-value surface problem and makes the maintained structure transparent.

The paper makes two research-level contributions. First, it makes the full common-value, two-population conditional distribution surface estimable under an explicit structure–flexibility trade-off. This capability permits coherent comparison over the continuum rather than a collection of unrelated local fits. Empirical margins and a fitted dependence representation are the enabling components, not the contribution headline. Second, the paper turns that fitted surface into simultaneous evidence about a uniform ordering. Building on contact-set methods for stochastic dominance \citep{LintonSongWhang2010} and directional inference for nonsmooth functionals \citep{FangSantos2019}, the analysis propagates uncertainty from estimated margins and dependence through a supremum whose least-favorable locations are unknown. This contribution is inseparable from the research question: without uniform calibration, the surface would remain a visualization or a collection of unadjusted pointwise contrasts rather than a defensible statement of dominance.

The theory validates the complete procedure under a correctly specified finite copula class, smoothness and trimming conditions, and a uniquely best candidate family. At a high level, it establishes uniform control of the estimated conditional-distribution contrast and consistent calibration of the one-sided dominance statistic. The simulations show why both the guarantee and its scope matter. Rejection probabilities generally rise as alternatives become easier to distinguish, while small-sample calibration remains sensitive to model selection and design. The numerical evidence therefore supports the procedure as a workable, structured option, but does not equate asymptotic validity with automatic robustness to misspecification or selection ambiguity.

An application to intergenerational income mobility illustrates the substantive payoff. The analysis compares adult family-income-rank distributions across parental-education groups at common childhood-income ranks, with childhood-income margins differing across groups. The high–low education comparison satisfies the paper's two-direction criterion after adjustment for multiplicity, whereas the adjacent high–middle and middle–low comparisons remain inconclusive. The result identifies bounded endpoint distributional separation without forcing a complete education gradient or a causal interpretation. The remainder of the paper develops the structured comparison, establishes its inferential properties, evaluates its finite-sample behavior, and reports the mobility application.
\section{Preliminaries}
\label{sec:preliminaries}

Conditional stochastic dominance is used to compare outcome distributions across populations after accounting for covariates. Existing conditional dominance procedures often focus on objects of the form $\Prob(Y\le y\mid X\le x)$ or rely on nonparametric conditional distribution estimators. These approaches are useful, but they do not directly address the pointwise conditional distribution $\Prob(Y\le y\mid X=x)$ when the conditioning variable is continuously distributed. In that case, the event $\{X=x\}$ has probability zero, and the conditional distribution is not obtained by dividing the joint distribution function by the marginal distribution function.

This paper studies pointwise conditional stochastic dominance through a copula derivative representation. If the joint distribution of $(X,Y)$ admits the copula representation $F(x,y)=C\{F_X(x),F_Y(y)\}$ and the relevant derivatives exist, then
\[
\Prob(Y\le y\mid X=x)=\frac{\partial_x F(x,y)}{f_X(x)}
=\partial_u C\{F_X(x),F_Y(y)\}.
\]
Thus, the comparison of pointwise conditional distributions can be transformed into a comparison of copula partial derivatives evaluated at the marginal probability indices. This transformation avoids direct estimation of conditional densities or bandwidth-dependent conditional distribution functions.

The proposed test combines this identification result with a semiparametric copula estimator. The marginal distributions are estimated nonparametrically by empirical distribution functions. The dependence structure is estimated by a finite collection of parametric copula families using canonical maximum pseudo-likelihood, and the candidate models are averaged by smooth AIC weights. Under a unique best-fitting copula condition, the model-averaged estimator has the same first-order behavior as the oracle estimator based on the selected copula family. The resulting test is a one-sided Kolmogorov--Smirnov statistic for the supremum of the difference between the estimated conditional distributions. Critical values are obtained by simulating the limiting Gaussian process using estimated influence functions.

\subsection{Pointwise Conditional Distribution}
\label{subsec:pointwise_conditional_distribution}

Let \(Z_{gi}=(X_{gi},Y_{gi})\), \(i=1,\ldots,n_g\), denote an i.i.d. sample from population \(g\in\{1,2\}\). The two samples are independent, and
\[
    \frac{n_1}{n_1+n_2}\to \lambda\in(0,1).
\]
Let \(F_g(x,y)\) be the joint distribution function of \((X_g,Y_g)\), with marginal distribution functions \(F_{gX}\) and \(F_{gY}\). Throughout, \(X_g\) is continuously distributed with density \(f_{gX}\). Inference is conducted on a compact trimmed set \(\Lambda_\varepsilon\) such that, for some fixed \(\varepsilon\in(0,1/2)\),
\[
    F_{gX}(x),\,F_{gY}(y)\in[\varepsilon,1-\varepsilon],
    \qquad
    (x,y)\in\Lambda_\varepsilon,\quad g=1,2.
\]
This trimming excludes boundary regions where copula derivatives may be unstable.

By Sklar's theorem, there exists a copula \(C_g\) such that
\[
    F_g(x,y)
    =
    C_g\!\left(F_{gX}(x),F_{gY}(y)\right).
\]
The object of interest is the pointwise conditional distribution of \(Y_g\) given \(X_g=x\),
\[
    H_g(y\mid x)
    =
    \Prob(Y_g\le y\mid X_g=x).
\]

When \(X_g\) is continuous, \(H_g(y\mid x)\) is not equal to \(F_g(x,y)/F_{gX}(x)\). Instead, if \(F_g(x,y)\) is differentiable with respect to \(x\) and \(f_{gX}(x)>0\), then
\[
    H_g(y\mid x)
    =
    \frac{\partial F_g(x,y)/\partial x}{f_{gX}(x)}.
\]
Using the copula representation and the chain rule,
\[
    \frac{\partial F_g(x,y)}{\partial x}
    =
    \partial_u C_g\!\left(F_{gX}(x),F_{gY}(y)\right)
    f_{gX}(x),
\]
where \(\partial_u C_g\) denotes the partial derivative of \(C_g\) with respect to its first argument. Therefore,
\[
    H_g(y\mid x)
    =
    \partial_u C_g\!\left(F_{gX}(x),F_{gY}(y)\right).
\]
This identity is the key observation underlying our approach: the pointwise conditional distribution can be represented by the first partial derivative of the copula evaluated at the marginal probability levels of \((x,y)\).

\subsection{Hypotheses and Test Statistics}
\label{subsec:hypotheses}

We say that population 1 dominates population 2 in the pointwise conditional stochastic dominance sense over \(\Lambda_\varepsilon\) if
\[
    H_1(y\mid x)\le H_2(y\mid x),
    \qquad
    \forall (x,y)\in\Lambda_\varepsilon.
\]
Equivalently, define the conditional dominance contrast
\[
    \Delta(x,y)
    =
    H_1(y\mid x)-H_2(y\mid x).
\]
The null hypothesis is
\[
    \mathcal H_0:
    \Delta(x,y)\le 0,
    \qquad
    \forall (x,y)\in\Lambda_\varepsilon,
\]
or equivalently,
\[
    \mathcal H_0:
    \sup_{(x,y)\in\Lambda_\varepsilon}\Delta(x,y)\le 0.
\]
The alternative hypothesis is
\[
    \mathcal H_1:
    \sup_{(x,y)\in\Lambda_\varepsilon}\Delta(x,y)>0.
\]

Using the copula derivative representation, the dominance contrast can be written as
\[
    \Delta(x,y)
    =
    \partial_u C_1\!\left(F_{1X}(x),F_{1Y}(y)\right)
    -
    \partial_u C_2\!\left(F_{2X}(x),F_{2Y}(y)\right).
\]
Thus, the testing problem is transformed into a comparison of two copula derivative functions evaluated at population-specific marginal probability levels.

The null hypothesis is one-sided and can be written as
\[
    \mathcal H_0:
    \sup_{(x,y)\in\Lambda_\varepsilon}\Delta(x,y)\le 0.
\]
Therefore, a natural test is based on the largest positive deviation of an estimated dominance contrast from zero.

Let \(\widehat H_g(y\mid x)\) denote an estimator of \(H_g(y\mid x)\), whose construction is described in the next subsection. Define
\[
    \widehat\Delta(x,y)
    =
    \widehat H_1(y\mid x)-\widehat H_2(y\mid x),
    \qquad
    (x,y)\in\Lambda_\varepsilon .
\]
Let
\[
    s_N
    =
    \sqrt{\frac{n_1n_2}{n_1+n_2}} .
\]
We consider the one-sided Kolmogorov--Smirnov statistic
\[
    T_N
    =
    s_N
    \sup_{(x,y)\in\Lambda_\varepsilon}
    \widehat\Delta(x,y).
\]
Large positive values of \(T_N\) indicate that the estimated conditional distribution of population 1 exceeds that of population 2 at some point in \(\Lambda_\varepsilon\), and hence provide evidence against pointwise conditional stochastic dominance. The asymptotic distribution of \(T_N\), the construction of the relevant contact set, and the critical value calculation are developed in the next section.

\subsection{Estimation Strategy}
\label{subsec:estimation_strategy}

We estimate \(H_g(y\mid x)\) using a semiparametric copula approach. Let
\[
    \mathcal M
    =
    \{C_1,\ldots,C_Q\}
\]
be a finite set of candidate parametric copula families. For each population \(g\), the method proceeds in three steps.

First, estimate the marginal distributions nonparametrically by the scaled empirical distribution functions
\[
    \widehat F_{gX}(x)
    =
    \frac{1}{n_g+1}
    \sum_{i=1}^{n_g}
    \mathbf 1(X_{gi}\le x),
    \qquad
    \widehat F_{gY}(y)
    =
    \frac{1}{n_g+1}
    \sum_{i=1}^{n_g}
    \mathbf 1(Y_{gi}\le y).
\]
The corresponding pseudo-observations are
\[
    \widehat U_{gi}
    =
    \widehat F_{gX}(X_{gi}),
    \qquad
    \widehat V_{gi}
    =
    \widehat F_{gY}(Y_{gi}).
\]

Second, for each candidate copula family \(C_\ell(\cdot,\cdot;\theta_\ell)\), estimate the copula parameter by canonical maximum pseudo-likelihood:
\[
    \widehat\theta_{g\ell}
    =
    \arg\max_{\theta_\ell\in\Theta_\ell}
    \sum_{i=1}^{n_g}
    \log c_\ell(\widehat U_{gi},\widehat V_{gi};\theta_\ell),
\]
where \(c_\ell\) is the density associated with \(C_\ell\).

Third, assign AIC weights to the candidate copulas. Let
\[
    \operatorname{AIC}_{g\ell}
    =
    -2
    \sum_{i=1}^{n_g}
    \log c_\ell(\widehat U_{gi},\widehat V_{gi};\widehat\theta_{g\ell})
    +
    2k_\ell,
\]
where \(k_\ell\) is the dimension of \(\theta_\ell\). Define
\[
    \widehat w_{g\ell}
    =
    \frac{
    \exp\!\left[
    -\frac{1}{2}
    \left\{
    \operatorname{AIC}_{g\ell}
    -
    \min_{1\le j\le Q}\operatorname{AIC}_{gj}
    \right\}
    \right]
    }{
    \sum_{m=1}^{Q}
    \exp\!\left[
    -\frac{1}{2}
    \left\{
    \operatorname{AIC}_{gm}
    -
    \min_{1\le j\le Q}\operatorname{AIC}_{gj}
    \right\}
    \right]
    }.
\]

The model-averaged estimator of the pointwise conditional distribution is
\[
    \widehat H_g(y\mid x)
    =
    \sum_{\ell=1}^{Q}
    \widehat w_{g\ell}
    \,
    \partial_u C_\ell
    \!\left(
    \widehat F_{gX}(x),
    \widehat F_{gY}(y);
    \widehat\theta_{g\ell}
    \right).
\]
Substituting this estimator into the contrast gives the feasible process
\[
    \widehat\Delta(x,y)
    =
    \widehat H_1(y\mid x)
    -
    \widehat H_2(y\mid x),
\]
which is used in the statistic \(T_N\) defined above.

\section{Asymptotic Theory}
\label{sec:asymptotic_theory}

This section establishes the asymptotic properties of the semiparametric
copula derivative estimator and the resulting supremum test statistic. The main
steps are as follows. First, we derive a uniform asymptotic linear
representation for \(\widehat H_g(y\mid x)\). Second, we combine the two-sample
weak convergence result with the directional differentiability of the supremum
functional to obtain the limiting distribution of \(T_N\). Finally, we construct
multiplier critical values based on estimated influence functions.

\subsection{Regularity Conditions}
\label{subsec:regularity_conditions}

For \(g\in\{1,2\}\), write
\[
    u_g(x)=F_{gX}(x),
    \qquad
    v_g(y)=F_{gY}(y).
\]
Let \(\ell_g^\ast\) denote the population copula model selected in the limit,
and let \(\theta_g\) be the corresponding parameter. For notational simplicity,
write
\[
    C_g(u,v)
    =
    C_{\ell_g^\ast}(u,v;\theta_g).
\]
All copula derivatives below are evaluated at
\((u_g(x),v_g(y);\theta_g)\). Specifically, define
\[
    C_{g,uu}(x,y)
    =
    \partial_{uu} C_{\ell_g^\ast}
    \{u_g(x),v_g(y);\theta_g\},
\]
\[
    C_{g,uv}(x,y)
    =
    \partial_{uv} C_{\ell_g^\ast}
    \{u_g(x),v_g(y);\theta_g\},
\]
and
\[
    C_{g,u\theta}(x,y)
    =
    \partial_{u\theta} C_{\ell_g^\ast}
    \{u_g(x),v_g(y);\theta_g\}.
\]

\begin{assumption}[Sampling, trimming, and smoothness]
\label{ass:basic_regular}
The two samples are independent, and
\(n_1/(n_1+n_2)\to\lambda\in(0,1)\). The set
\(\Lambda_\varepsilon\) is compact and satisfies
\(F_{gX}(x),F_{gY}(y)\in[\varepsilon,1-\varepsilon]\) for all
\((x,y)\in\Lambda_\varepsilon\) and \(g=1,2\). The marginal distributions are
continuous, and the density \(f_{gX}\) is bounded away from zero on the relevant
support. For each \(g\), the derivatives
\(\partial_u C_{\ell_g^\ast}\), \(\partial_{uu} C_{\ell_g^\ast}\),
\(\partial_{uv} C_{\ell_g^\ast}\), and
\(\partial_{u\theta} C_{\ell_g^\ast}\) exist, are uniformly bounded, and are
uniformly continuous on the trimmed domain.
\end{assumption}

\begin{assumption}[Candidate copulas and first-stage estimation]
\label{ass:copula_first_stage}
The candidate set $\mathcal M=\{C_1,\ldots,C_Q\}$ is finite and fixed.
For each population $g$ and each candidate family $\ell$, let
\[
    M_{g\ell}(\theta)
    =
    E\{\log c_\ell(U_g,V_g;\theta)\}.
\]
The corresponding normalized sample pseudo-likelihood criterion is
\[
    \widehat M_{g\ell}(\theta)
    =
    \frac{1}{n_g}
    \sum_{i=1}^{n_g}
    \log c_\ell(\widehat U_{gi},\widehat V_{gi};\theta).
\]
The criterion $M_{g\ell}$ has a unique maximizer
$\theta_{g\ell}^{\ast}$, and $\widehat M_{g\ell}$ satisfies
\[
    \sup_{\theta\in\Theta_\ell}
    \left|
    \widehat M_{g\ell}(\theta)-M_{g\ell}(\theta)
    \right|
    \overset{p}{\longrightarrow}0.
\]

The true copula belongs to the candidate set and is represented by a
unique family $\ell_g^\ast$. Moreover, for every
$\ell\ne\ell_g^\ast$,
\[
    \kappa_{g\ell}
    =
    M_{g\ell_g^\ast}(\theta_g)
    -
    M_{g\ell}(\theta_{g\ell}^{\ast})
    >0.
\]

For the limiting family $\ell_g^\ast$, $\theta_g$ is an interior point of
$\Theta_{\ell_g^\ast}$. On a neighborhood of $\theta_g$, the copula
log-density score and its derivatives with respect to the parameter and the two
copula arguments are measurable, continuously differentiable in their relevant
arguments, and dominated by square-integrable envelopes. The corresponding
local score and derivative classes are $P$-Donsker, the score map is
continuously differentiable under marginal-distribution perturbations, and the
negative expected parameter Hessian is nonsingular.

The canonical maximum pseudo-likelihood estimator in family
$\ell_g^\ast$ satisfies
\[
    \sqrt{n_g}(\widehat\theta_g-\theta_g)
    =
    \frac{1}{\sqrt{n_g}}
    \sum_{i=1}^{n_g}\psi_g(Z_{gi})
    +o_p(1),
\]
where $E\{\psi_g(Z_g)\}=0$,
$E\|\psi_g(Z_g)\|^2<\infty$, and
$\operatorname{Var}\{\psi_g(Z_g)\}$ is positive definite.
\end{assumption}
\begin{remark}[Influence function of the CMPL estimator]
\label{rem:psi_g}
Let
\[
    U_g=F_{gX}(X_g),
    \qquad
    V_g=F_{gY}(Y_g),
\]
and let $(U_g',V_g')$ be an independent copy of $(U_g,V_g)$. 
For the limiting copula family $\ell_g^\ast$, define the
log copula density
\[
    \ell_g(u,v;\theta)
    =
    \log c_{\ell_g^\ast}(u,v;\theta).
\]
The score evaluated at the true parameter is
\[
    s_g(u,v)
    =
    \left.
    \nabla_\theta
    \ell_g(u,v;\theta)
    \right|_{\theta=\theta_g},
\]
and the sensitivity matrix is
\[
    J_g
    =
    -E\left[
    \left.
    \nabla_{\theta\theta^\top}^{2}
    \ell_g(U_g,V_g;\theta)
    \right|_{\theta=\theta_g}
    \right].
\]
For $u_0,v_0\in(0,1)$, define the marginal-estimation correction
functions
\[
\begin{aligned}
    r_{gX}(u_0)
    &=
    E\left[
    \partial_u s_g(U_g',V_g')
    \left\{
    \mathbf 1(u_0\le U_g')-U_g'
    \right\}
    \right],\\
    r_{gY}(v_0)
    &=
    E\left[
    \partial_v s_g(U_g',V_g')
    \left\{
    \mathbf 1(v_0\le V_g')-V_g'
    \right\}
    \right].
\end{aligned}
\]
Then the influence function in
Assumption~\ref{ass:copula_first_stage} is
\[
    \psi_g(Z_g)
    =
    J_g^{-1}
    \left[
    s_g(U_g,V_g)
    +
    r_{gX}(U_g)
    +
    r_{gY}(V_g)
    \right].
\]
The first term is the usual copula score, while the last two terms account
for replacing the unknown marginal distributions by empirical distribution
functions. The representation requires the combined corrected score to be
well defined and square integrable and $J_g$ to be nonsingular.
\end{remark}

Under Assumption~\ref{ass:copula_first_stage}, the AIC weights concentrate on
the uniquely identified copula model. Hence the model-averaged estimator is
first-order equivalent to the oracle estimator based on
\(\ell_g^\ast\).

\begin{lemma}[Oracle equivalence of the model-averaged estimator]
\label{lem:oracle_equivalence}
Under Assumptions~\ref{ass:basic_regular} and
\ref{ass:copula_first_stage},
\[
    \sup_{(x,y)\in\Lambda_\varepsilon}
    \left|
    \widehat H_g(y\mid x)
    -
    \partial_u C_{\ell_g^\ast}
    \{\widehat F_{gX}(x),\widehat F_{gY}(y);\widehat\theta_g\}
    \right|
    =
    o_p(n_g^{-1/2}),
\]
for \(g=1,2\).
\end{lemma}

\subsection{Uniform Linear Expansion}
\label{subsec:uniform_linear_expansion}

Define the influence function of \(\widehat H_g(y\mid x)\) at
\(z=(x,y)\) by
\[
\begin{aligned}
    \phi_{g,z}(Z_g)
    =
    &\ C_{g,uu}(x,y)
    \left\{\mathbf 1(X_g\le x)-F_{gX}(x)\right\}
    \\
    &+
    C_{g,uv}(x,y)
    \left\{\mathbf 1(Y_g\le y)-F_{gY}(y)\right\}
    \\
    &+
    C_{g,u\theta}(x,y)^\top
    \psi_g(Z_g).
\end{aligned}
\]
The first two terms capture the effect of estimating the marginal
distributions, while the last term captures the effect of estimating the copula
parameter.

\begin{theorem}[Uniform asymptotic linearity]
\label{thm:linear_h}
Under Assumptions~\ref{ass:basic_regular} and
\ref{ass:copula_first_stage}, for \(g=1,2\),
\[
    \sup_{(x,y)\in\Lambda_\varepsilon}
    \left|
    \sqrt{n_g}
    \{\widehat H_g(y\mid x)-H_g(y\mid x)\}
    -
    \frac{1}{\sqrt{n_g}}
    \sum_{i=1}^{n_g}
    \phi_{g,(x,y)}(Z_{gi})
    \right|
    =
    o_p(1).
\]
Consequently,
\[
    \sqrt{n_g}
    \{\widehat H_g-H_g\}
    \rightsquigarrow
    \mathbb G_g
    \quad
    \text{in }
    \ell^\infty(\Lambda_\varepsilon),
\]
where \(\mathbb G_g\) is a tight mean-zero Gaussian process with covariance
kernel
\[
    \Sigma_g(z,z')
    =
    E\{\phi_{g,z}(Z_g)\phi_{g,z'}(Z_g)\}.
\]
\end{theorem}

The empirical-process conditions required for
Theorem~\ref{thm:linear_h} are not imposed as separate high-level assumptions.
They follow from the VC property of lower-orthant indicator classes, the
bounded smoothness of the copula derivative weights on the trimmed domain, and
the finite-dimensional asymptotic linear representation of
\(\widehat\theta_g\).

\subsection{Asymptotic Properties of the Statistic}
\label{subsec:two_sample_supremum_limit}

Before stating the result, define
\[
    \mathbb G_\Delta(z)
    =
    \sqrt{1-\lambda}\,\mathbb G_1(z)
    -
    \sqrt{\lambda}\,\mathbb G_2(z),
    \qquad
    z\in\Lambda_\varepsilon,
\]
where \(\mathbb G_1\) and \(\mathbb G_2\) are independent copies of the
population-specific Gaussian limits in Theorem~\ref{thm:linear_h}.  The limit process has covariance kernel
\[
    \Sigma_\Delta(z,z')
    =
    (1-\lambda)\Sigma_1(z,z')
    +
    \lambda\Sigma_2(z,z').
\]
Under the continuity and Donsker conditions used in
Theorem~\ref{thm:linear_h}, the process
\(\mathbb G_\Delta\) admits a version with almost surely uniformly
continuous sample paths on the compact set
\(\Lambda_\varepsilon\). Also define
\[
    \Gamma^\ast(\Delta)
    =
    \left\{
    z\in\Lambda_\varepsilon:
    \Delta(z)
    =
    \sup_{\tilde z\in\Lambda_\varepsilon}\Delta(\tilde z)
    \right\}.
\]
Under the boundary null, \(\sup_{z\in\Lambda_\varepsilon}\Delta(z)=0\), this set
coincides with the contact set
\[
    \Gamma(\Delta)
    =
    \{z\in\Lambda_\varepsilon:\Delta(z)=0\}.
\]

\begin{theorem}[Two-sample weak convergence]
\label{thm:two_sample_sup_limit}
Under Assumptions~\ref{ass:basic_regular} and
\ref{ass:copula_first_stage},
\[
    s_N
    \{\widehat\Delta-\Delta\}
    \rightsquigarrow
    \mathbb G_\Delta
    \quad
    \text{in }
    \ell^\infty(\Lambda_\varepsilon).
\]
Moreover,
\[
    T_N
    -
    s_N
    \sup_{z\in\Lambda_\varepsilon}
    \Delta(z)
    \rightsquigarrow
    \sup_{z\in\Gamma^\ast(\Delta)}
    \mathbb G_\Delta(z).
\]

\end{theorem}

\begin{remark}[Local alternatives]
\label{rem:local_power}
Let \(\Delta_0\in C(\Lambda_\varepsilon)\) be a boundary-null contrast satisfying
\[
    \Delta_0(z)\le0
    \quad\text{for all }z\in\Lambda_\varepsilon,
    \qquad
    \sup_{z\in\Lambda_\varepsilon}\Delta_0(z)=0,
\]
and consider a sequence of local perturbations such that, uniformly over
\(\Lambda_\varepsilon\),
\[
    \Delta_N(z)
    =
    \Delta_0(z)
    +
    s_N^{-1}h(z)
    +
    o(s_N^{-1}),
\]
where \(h\in C(\Lambda_\varepsilon)\). Suppose the conditions of
Theorem~\ref{thm:two_sample_sup_limit} hold uniformly along this sequence and
the centered contrast process has the same Gaussian weak limit
\(\mathbb G_\Delta\). Then
\[
    T_N
    \rightsquigarrow
    \sup_{z\in\Gamma(\Delta_0)}
    \left\{
    h(z)+\mathbb G_\Delta(z)
    \right\},
\]
where
\(
\Gamma(\Delta_0)
=
\{z\in\Lambda_\varepsilon:\Delta_0(z)=0\}
\)
is the contact set of the limiting null.

Let \(c_{1-\alpha}(\Delta_0)\) denote the \((1-\alpha)\)-quantile of
\(\sup_{z\in\Gamma(\Delta_0)}\mathbb G_\Delta(z)\). If its distribution is
continuous at this quantile, the local asymptotic rejection probability is
\[
\Prob\left[
    \sup_{z\in\Gamma(\Delta_0)}
    \left\{h(z)+\mathbb G_\Delta(z)\right\}
    >c_{1-\alpha}(\Delta_0)
    \right].
\]
Thus, first-order local power is determined by the restriction of \(h\) to the
contact set. Positive drift on that set can increase rejection probability,
whereas directions satisfying
\(\sup_{z\in\Gamma(\Delta_0)}h(z)\le0\) do not constitute an outward
first-order perturbation of the null at its binding points.
\end{remark}

Let \(c_{1-\alpha}\) denote the
\((1-\alpha)\)-quantile of 
\( \sup_{z\in\Gamma(\Delta)}\mathbb G_\Delta(z)\)
when \(\sup_{z\in\Lambda_\varepsilon}\Delta(z)=0\), and set
\(c_{1-\alpha}=0\) when \(\sup_{z\in\Lambda_\varepsilon}\Delta(z)<0\).
When the null is on the boundary, we assume that 
\(\sup_{z\in\Gamma(\Delta)}\mathbb G_\Delta(z)\) is continuous at
\(c_{1-\alpha}\). 
The decision rule is to reject \(\mathcal{H}_0\) whenever
\[
    T_N>c_{1-\alpha}.
\]

\begin{theorem}[Oracle size control and consistency]
\label{thm:oracle_test_validity}
Suppose Assumptions~\ref{ass:basic_regular} and
\ref{ass:copula_first_stage} hold. Then, under \(\mathcal{H}_0\),
\[
    \lim_{n_1,n_2\to\infty}
    \Prob(T_N>c_{1-\alpha})
    \le \alpha .
\]
Moreover, under  \(\mathcal{H}_1\),
\[
    \Prob(T_N>c_{1-\alpha})\to1 .
\]
\end{theorem}

\begin{remark}[Feasible size control]
\label{rem:feasible_size}
Theorem~\ref{thm:oracle_test_validity} states size control for the oracle
critical value \(c_{1-\alpha}\). In practice, the critical value is estimated by
the multiplier or bootstrap procedures of Section~\ref{sec:critical_values},
whose consistency is established in
Theorems~\ref{thm:multiplier} and~\ref{thm:bootstrap_cv_consistency}.
Combining the two results, under \(\mathcal{H}_0\),
\[
    \limsup_{n_1,n_2\to\infty}
    \Prob\left\{T_N>\widehat c_{1-\alpha}^{(m)}\right\}
    \le \alpha,
\]
with equality along the boundary null
\(\sup_{z\in\Lambda_\varepsilon}\Delta(z)=0\) when the limiting distribution of
\(\sup_{z\in\Gamma(\Delta)}\mathbb G_\Delta(z)\) is continuous at
\(c_{1-\alpha}\). The analogous statement holds for \(\widehat c_{1-\alpha}^{*}\).
\end{remark}
The preceding results concern a prespecified direction.  In applications,
however, we test both directions so that a dominance conclusion requires not
only failure to reject the proposed ordering but also rejection of the reverse
ordering. For
\(j,k\in\{1,2\}\), \(j\ne k\), define
\[
    \Delta_{12}(z)=\Delta(z),
    \qquad
    \Delta_{21}(z)=-\Delta(z),
\]
and, correspondingly,
\[
    \widehat\Delta_{12}(z)=\widehat\Delta(z),
    \qquad
    \widehat\Delta_{21}(z)=-\widehat\Delta(z).
\]
The two directional null hypotheses are
\[
    \mathcal H_{jk}:
    \Delta_{jk}(z)\leq0
    \quad\text{for all }z\in\Lambda_\varepsilon,
\]
with test statistics
\[
    T_{jk,N}
    =
    s_N\sup_{z\in\Lambda_\varepsilon}
    \widehat\Delta_{jk}(z).
\]
Let \(\widehat c_{jk,1-\alpha}\) denote the corresponding feasible multiplier
or bootstrap critical value.  The two-direction rule declares population 1
to dominate population 2 when the event
\[
    \mathcal D_{12,N}
    =
    \left\{
    T_{12,N}\leq\widehat c_{12,1-\alpha},
    \quad
    T_{21,N}>\widehat c_{21,1-\alpha}
    \right\},
\]
occurs, and defines \(\mathcal D_{21,N}\) analogously.

The following result follows directly from
\[
    \mathcal D_{12,N}
    \subseteq
    \{T_{21,N}>\widehat c_{21,1-\alpha}\}
\]
and
\[
    \Prob(\mathcal D_{12,N})
    \geq
    1
    -
    \Prob(T_{12,N}>\widehat c_{12,1-\alpha})
    -
    \Prob(T_{21,N}\leq\widehat c_{21,1-\alpha}),
\]
together with directional size control and fixed-alternative consistency.  No
independence assumption between the two directional tests is required.

\begin{corollary}[Two-direction operating characteristics]
\label{cor:two_direction_rule}
Suppose that both directional tests satisfy feasible size control and
fixed-alternative consistency.

\begin{enumerate}
    \item If \(\mathcal H_{21}\) holds, then
    \[
        \limsup_{N\to\infty}
        \Prob(\mathcal D_{12,N})
        \leq\alpha.
    \]

    \item If \(\mathcal H_{12}\) holds and \(\mathcal H_{21}\) is false under
    a fixed alternative, then
    \[
        \liminf_{N\to\infty}
        \Prob(\mathcal D_{12,N})
        \geq1-\alpha.
    \]
    If, in addition,
    \[
        \sup_{z\in\Lambda_\varepsilon}\Delta_{12}(z)\leq-\eta
    \]
    for some \(\eta>0\), then
    \[
        \Prob(\mathcal D_{12,N})\longrightarrow1.
    \]

    \item If both \(\mathcal H_{12}\) and \(\mathcal H_{21}\) are false under
    fixed alternatives, then both directional tests reject with probability
    approaching one and
    \[
        \Prob(
        \mathcal D_{12,N}\cup\mathcal D_{21,N}
        )
        \longrightarrow0.
    \]
\end{enumerate}

The conclusions hold symmetrically after reversing the population labels.
If Holm-adjusted rejection decisions are used to define
\(\mathcal D_{jk,N}\) across a finite family of directional tests, the
probability of making any directional declaration for which the reverse null
is true is asymptotically bounded by the familywise level \(\alpha\).  In
particular, under equality, the probability of declaring either direction is
asymptotically bounded by \(\alpha\).
\end{corollary}

The fixed-alternative qualification is important.  Under
\(s_N^{-1}\)-local deviations from the boundary, the corresponding
classification probabilities are governed by
Remark~\ref{rem:local_power} and need not converge to either zero or one.

\begin{remark}[Discretized evaluation region]
\label{rem:grid}
Let \(\Lambda_N\subseteq\Lambda_\varepsilon\) be a
\(J_N\times J_N\) grid with mesh size
\(h_N=\sup_{z\in\Lambda_\varepsilon}d(z,\Lambda_N)\). Suppose that the relevant
maximizers are interior, \(\Delta\) is twice continuously differentiable with
bounded Hessian near them, and the contrast and resampling processes are
stochastically equicontinuous at scale \(h_N\). If \(s_Nh_N^2\to0\), then
\[
    s_N\left\{
    \sup_{z\in\Lambda_\varepsilon}\widehat\Delta(z)
    -\max_{z\in\Lambda_N}\widehat\Delta(z)
    \right\}=o_p(1),
\]
so grid evaluation does not affect the first-order limiting distribution or
resampling calibration. For a regular grid, \(h_N=O(J_N^{-1})\), so it suffices
that \(s_N/J_N^2\to0\), or equivalently \(J_N/N^{1/4}\to\infty\) when
\(s_N\asymp\sqrt N\). Under standard quantile regularity conditions, the same
mesh order holds in probability for the pooled empirical-quantile grid.

The simulations and application use \(J=25\) as a finite-sample choice. For
all reported sample sizes and pairwise comparisons, \(s_N/25^2\le0.036\).
Nevertheless, a grid held fixed as \(N\to\infty\) need not be asymptotically
equivalent to the continuum and may miss violations between grid points.
\end{remark}

\section{Choice of Critical Values}
\label{sec:critical_values}
\subsection{Multiplier Critical Values}
\label{subsec:multiplier_cv}

Throughout this section, the notation \(\rightsquigarrow_\xi\) (respectively
\(\rightsquigarrow_B\)) denotes weak convergence of a random element generated
by the multiplier variables \(\{\xi_{gi}\}\) (respectively the bootstrap weights
\(\{B_{gi}^{*}\}\)), conditional on the observed data, in probability.

The limiting distribution in Theorem~\ref{thm:two_sample_sup_limit} depends on
the unknown contact set and the unknown covariance structure of
\(\mathbb G_\Delta\). We approximate it by an analytic influence-function
(AIF) multiplier procedure that simulates the limiting Gaussian process using
estimated influence functions. We refer to this procedure as the AIF method
below.

Let \(\widehat\ell_g\) denote the AIC-selected copula family for
population \(g\), and write
\[
    \widehat\theta_g
    =
    \widehat\theta_{g\widehat\ell_g}.
\]
For the selected family, define the score
\[
    s_{g}(u,v;\theta)
    =
    \partial_{\theta}
    \log c_{\widehat\ell_g}(u,v;\theta),
\]
and let
\[
    \widehat s_{gi}
    =
    s_g(\widehat U_{gi},\widehat V_{gi};
    \widehat\theta_g).
\]
The estimated sensitivity matrix is
\[
    \widehat J_g
    =
    -
    \frac{1}{n_g}
    \sum_{j=1}^{n_g}
    \partial_\theta
    s_g(\widehat U_{gj},\widehat V_{gj};
    \widehat\theta_g).
\]

To account for estimation of the two marginal distributions, define
the observation-specific correction terms
\[
\begin{aligned}
    \widehat r_{gX,i}
    &=
    \frac{1}{n_g}
    \sum_{j=1}^{n_g}
    \partial_u
    s_g(\widehat U_{gj},\widehat V_{gj};
    \widehat\theta_g)
    \left\{
    \mathbf 1(\widehat U_{gi}\le\widehat U_{gj})
    -
    \widehat U_{gj}
    \right\},\\
    \widehat r_{gY,i}
    &=
    \frac{1}{n_g}
    \sum_{j=1}^{n_g}
    \partial_v
    s_g(\widehat U_{gj},\widehat V_{gj};
    \widehat\theta_g)
    \left\{
    \mathbf 1(\widehat V_{gi}\le\widehat V_{gj})
    -
    \widehat V_{gj}
    \right\}.
\end{aligned}
\]
Let
\[
    \widehat e_{gi}
    =
    \widehat s_{gi}
    +
    \widehat r_{gX,i}
    +
    \widehat r_{gY,i},
    \qquad
    \overline e_g
    =
    \frac{1}{n_g}
    \sum_{j=1}^{n_g}\widehat e_{gj}.
\]
The plug-in influence-function estimate is
\[
    \widehat\psi_{gi}
    =
    \widehat J_g^{-1}
    \left(
    \widehat e_{gi}-\overline e_g
    \right).
    \label{eq:theta_if_plugin}
\]
For the independence copula, which has no estimated dependence
parameter, the corresponding parameter influence term is absent.
The score derivatives and the sensitivity matrix are evaluated
numerically when closed-form expressions are unavailable.

Using \(\widehat\psi_{gi}\), define
\[
\begin{aligned}
    \widehat\phi_{g,z}(Z_{gi})
    =
    &\ \widehat C_{g,uu}(z)
    \left\{\mathbf 1(X_{gi}\le x)-\widehat F_{gX}(x)\right\}
    \\
    &+
    \widehat C_{g,uv}(z)
    \left\{\mathbf 1(Y_{gi}\le y)-\widehat F_{gY}(y)\right\}
    \\
    &+
    \widehat C_{g,u\theta}(z)^\top
    \widehat\psi_{gi},
\end{aligned}
\]
where the copula derivatives are evaluated at
\[
    \left(
    \widehat F_{gX}(x),
    \widehat F_{gY}(y);
    \widehat\theta_{g\widehat\ell_g}
    \right).
\]

Let \(\{\xi_{gi}\}\) be independent multipliers, independent of the data, with
\(E(\xi_{gi})=0\), \(E(\xi_{gi}^2)=1\), and finite \(2+\delta\) moment for some
\(\delta>0\). Define
\[
    \widehat{\mathbb G}_N^{(m)}(z)
    =
    \sqrt{1-\widehat\lambda}
    \frac{1}{\sqrt{n_1}}
    \sum_{i=1}^{n_1}
    \xi_{1i}\widehat\phi_{1,z}(Z_{1i})
    -
    \sqrt{\widehat\lambda}
    \frac{1}{\sqrt{n_2}}
    \sum_{i=1}^{n_2}
    \xi_{2i}\widehat\phi_{2,z}(Z_{2i}),
\]
where \(\widehat\lambda=n_1/(n_1+n_2)\).

In the simulations and empirical application, we use i.i.d. Mammen
two-point multipliers \citep{Mammen1993},
\[
\xi=
\begin{cases}
(1-\sqrt{5})/2,
& \text{with probability }(\sqrt{5}+1)/(2\sqrt{5}),\\
(1+\sqrt{5})/2,
& \text{with probability }(\sqrt{5}-1)/(2\sqrt{5}),
\end{cases}
\]
such that $E[\xi]=0$, $E[\xi^2]=1$, and $E[\xi^3]=1$. The multipliers are generated independently across observations,
populations, and resampling draws.

Let \(a_N\to\infty\) and \(a_N/s_N\to0\). Estimate the contact set by
\begin{equation}
\label{eq:contact_set}
    \widehat\Gamma_N
    =
    \left\{
    z\in\Lambda_\varepsilon:
    s_N\widehat\Delta(z)>-a_N
    \right\}.
\end{equation}
The multiplier statistic is
\[
    T_N^{(m)}
    =
    \sup_{z\in\widehat\Gamma_N}
    \widehat{\mathbb G}_N^{(m)}(z),
\]
with the convention that \(T_N^{(m)}=0\) if
\(\widehat\Gamma_N=\varnothing\). Let \(\widehat c_{1-\alpha}^{\,(m)}\) be the
conditional \((1-\alpha)\)-quantile of \(T_N^{(m)}\) given the data.

\begin{theorem}[Consistency of multiplier critical values]
\label{thm:multiplier}
Suppose Assumptions~\ref{ass:basic_regular} and
\ref{ass:copula_first_stage} hold. Let the multipliers be independent of the
data, have mean zero, variance one, and finite \(2+\delta\) moment for some
\(\delta>0\). Under the boundary null,
\[
    T_N^{(m)}
    \rightsquigarrow_\xi
    \sup_{z\in\Gamma(\Delta)}
    \mathbb G_\Delta(z)
    \quad
    \text{in probability}.
\]

If the distribution of
\(\sup_{z\in\Gamma(\Delta)}\mathbb G_\Delta(z)\) is continuous at its
\((1-\alpha)\)-quantile \(c_{1-\alpha}\), then
\[
    \widehat c_{1-\alpha}^{\,(m)}
    \overset{p}{\longrightarrow}
    c_{1-\alpha}.
\]
\end{theorem}

\begin{remark}[Nonlinear AIF implementation]
\label{rem:nonlinear}
Theorem~\ref{thm:multiplier} establishes the first-order validity of the
influence-function multiplier procedure under the maintained
unique-best-family condition. In finite samples, however, several candidate
copula families may receive similar empirical support, so variation in the
model-selection step can remain non-negligible. To partially accommodate this
additional source of finite-sample uncertainty, our implementation augments the
first-order multiplier approximation with a nonlinear perturbation of the AIC
weights.

For each multiplier draw, we perturb both the candidate conditional CDFs by
their estimated influence functions and the candidate log-likelihoods by their
centered observation-level contributions, and then recompute the AIC weights
through the exact softmax map. In compact form,
\begin{align*}
    \widehat H_g^{(b)}(z)
    &=
    \sum_{\ell=1}^{Q}\widehat w_{g\ell}^{(b)}
    \left\{\widehat h_{g\ell}(z)
    +\frac{1}{n_g}\sum_{i=1}^{n_g}
    \xi_{gi}^{(b)}\widehat\phi_{g\ell,z}(Z_{gi})\right\},
    \\
    \widehat w_{g\ell}^{(b)}
    &\propto
    \exp\!\left\{-\frac{1}{2}\operatorname{AIC}_{g\ell}
    +\sum_{i=1}^{n_g}\xi_{gi}^{(b)}\widehat q_{g\ell i}\right\},
\end{align*}
where \(\widehat q_{g\ell i}\) is the centered fitted log-density contribution.
Thus, the relative empirical support for the candidate families is allowed to
vary across multiplier draws, rather than being fixed at its full-sample value.

This nonlinear perturbation is a finite-sample stabilization device rather
than an additional first-order ingredient of the asymptotic theory. Under the
maintained unique-best-family condition, the AIC weight on the limiting family
converges to one, while the weights on the remaining families vanish
exponentially. The nonlinear weight perturbation is therefore asymptotically
inactive at first order and does not alter either the oracle theory or the
validity result in Theorem~\ref{thm:multiplier}.
\end{remark}

\subsection{Bootstrap Critical Values}
\label{subsec:bootstrap_cv}

The multiplier critical value in Section~\ref{subsec:multiplier_cv} is based on the
estimated influence-function representation and therefore avoids re-estimating the
copula parameters. As an alternative, we also consider a full weighted bootstrap
procedure that recomputes the marginal distributions, copula parameters, AIC
weights, and the conditional distribution estimator in each bootstrap replication.
To avoid confusion with the multiplier variables \(\xi_{gi}\) used above, we denote
the bootstrap frequency weights by \(B_{gi}^{*}\).

We impose the following condition on the bootstrap weights.

\begin{assumption}[Bootstrap weights]
\label{ass:bootstrap_weights}
For each population \(g\in\{1,2\}\), let
\(B_g^{*}=(B_{g1}^{*},\ldots,B_{gn_g}^{*})^\top\) be generated independently of the
data and independently across populations. The vector \(B_g^{*}\) is exchangeable,
\(B_{gi}^{*}\ge 0\), and \(\sum_{i=1}^{n_g}B_{gi}^{*}=n_g\). Moreover, for a generic
component \(B_{g1}^{*}\),
\[
    \limsup_{n_g\to\infty}\|B_{g1}^{*}\|_{2,1}<\infty,
    \qquad
    \lim_{\lambda\to\infty}\limsup_{n_g\to\infty}
    \sup_{t\ge \lambda}t^2\Prob(B_{g1}^{*}>t)=0,
\]
where \(\|B_{g1}^{*}\|_{2,1}=\int_0^\infty \Prob(B_{g1}^{*}\ge u)^{1/2}\,du\). Finally,
\[
    \frac{1}{n_g}\sum_{i=1}^{n_g}(B_{gi}^{*}-1)^2
    \overset{p}{\longrightarrow}1 .
\]
These conditions are satisfied, for example, by the multinomial bootstrap weights
\[
    (B_{g1}^{*},\ldots,B_{gn_g}^{*})
    \sim
    \mathrm{Multinomial}\{n_g;(1/n_g,\ldots,1/n_g)\}.
\]
\end{assumption}

For each bootstrap replication, compute the weighted empirical marginal
distribution functions
\[
    \widehat F_{gX}^{*}(x)
    =
    \frac{1}{n_g+1}\sum_{i=1}^{n_g}B_{gi}^{*}1(X_{gi}\le x),
    \qquad
    \widehat F_{gY}^{*}(y)
    =
    \frac{1}{n_g+1}\sum_{i=1}^{n_g}B_{gi}^{*}1(Y_{gi}\le y).
\]
The corresponding bootstrap pseudo-observations are
\[
    \widehat U_{gi}^{*}=\widehat F_{gX}^{*}(X_{gi}),
    \qquad
    \widehat V_{gi}^{*}=\widehat F_{gY}^{*}(Y_{gi}).
\]
For each candidate copula family \(C_\ell(\cdot,\cdot;\theta_\ell)\), define the
normalized weighted pseudo-likelihood criterion
\[
    \widehat M_{g\ell}^{*}(\theta)
    =
    \frac{1}{n_g}
    \sum_{i=1}^{n_g}
    B_{gi}^{*}
    \log c_\ell(\widehat U_{gi}^{*},\widehat V_{gi}^{*};\theta)
\]
and its maximizer
\[
    \widehat\theta_{g\ell}^{*}
    =
    \arg\max_{\theta_\ell\in\Theta_\ell}
    \widehat M_{g\ell}^{*}(\theta_\ell).
\]
The bootstrap AIC criterion is
\[
    AIC_{g\ell}^{*}
    =
    -2\sum_{i=1}^{n_g}
    B_{gi}^{*}
    \log c_\ell(\widehat U_{gi}^{*},\widehat V_{gi}^{*};\widehat\theta_{g\ell}^{*})
    +2k_\ell,
\]
and the corresponding bootstrap AIC weights are
\[
    \widehat w_{g\ell}^{*}
    =
    \frac{
    \exp[-\{AIC_{g\ell}^{*}-\min_{1\le j\le Q}AIC_{gj}^{*}\}/2]
    }{
    \sum_{m=1}^Q
    \exp[-\{AIC_{gm}^{*}-\min_{1\le j\le Q}AIC_{gj}^{*}\}/2]
    } .
\]
The bootstrap estimator of the pointwise conditional distribution is
\[
    \widehat H_g^{*}(y\mid x)
    =
    \sum_{\ell=1}^Q
    \widehat w_{g\ell}^{*}
    \partial_u C_\ell
    \{\widehat F_{gX}^{*}(x),\widehat F_{gY}^{*}(y);
    \widehat\theta_{g\ell}^{*}\}.
\]

\begin{assumption}[Bootstrap criterion regularity]
\label{ass:bootstrap_criterion}
For each population \(g\in\{1,2\}\), conditionally on the observed data and in
probability,
\[
    \max_{1\le\ell\le Q}
    \sup_{\theta\in\Theta_\ell}
    \left|
    \widehat M_{g\ell}^{*}(\theta)
    -
    \widehat M_{g\ell}(\theta)
    \right|
    =
    o_p^{*}(1).
\]
\end{assumption}

\begin{remark}[Regularity of the bootstrap criterion]
Assumption~\ref{ass:bootstrap_criterion} is the conditional bootstrap analogue
of the uniform convergence requirement in
Assumption~\ref{ass:copula_first_stage}. It follows from the standard uniform
law for exchangeably weighted empirical processes and stochastic
equicontinuity of the smooth pseudo-likelihood criteria under empirical-margin
perturbations \citep{PraestgaardWellner1993Exchangeably,Tsukahara2005}. With a
fixed finite candidate set, compact parameter spaces away from singular
parameter values, and the usual integrable-envelope conditions for rank-based
canonical maximum pseudo-likelihood, this is a standard regularity condition
for the commonly used parametric copula families.
\end{remark}

Define the bootstrap contrast
\[
    \widehat\Delta^{*}(x,y)
    =
    \widehat H_1^{*}(y\mid x)-\widehat H_2^{*}(y\mid x).
\]
Use the same estimated contact set 
in (\ref{eq:contact_set}), the bootstrap statistic is
\[
    T_N^{*}
    =
    \sup_{z\in\widehat\Gamma_N}
    s_N\{\widehat\Delta^{*}(z)-\widehat\Delta(z)\},
\]
with the convention that \(T_N^{*}=0\) if \(\widehat\Gamma_N=\emptyset\). Let
\(\widehat c_{1-\alpha}^{\,*}\) be the conditional \((1-\alpha)\)-quantile of
\(T_N^{*}\) given the data.

\begin{theorem}[Consistency of bootstrap critical values]
\label{thm:bootstrap_cv_consistency}
Suppose Assumptions~\ref{ass:basic_regular},
\ref{ass:copula_first_stage}, \ref{ass:bootstrap_weights}, and
\ref{ass:bootstrap_criterion} hold. Under the boundary null
\(\sup_{z\in\Lambda_\varepsilon}\Delta(z)=0\),

\[
    T_N^{*}
    \rightsquigarrow_{B}
    \sup_{z\in\Gamma(\Delta)}\mathbb{G}_\Delta(z)
    \quad\text{in probability}.
\]
 If the distribution of
\(\sup_{z\in\Gamma(\Delta)}\mathbb{G}_\Delta(z)\) is continuous at its
\((1-\alpha)\)-quantile \(c_{1-\alpha}\), then
\[
    \widehat c_{1-\alpha}^{\,*}
    \overset{p}{\longrightarrow}
    c_{1-\alpha}.
\]

\end{theorem}

\subsection{Model-selection uncertainty and computational trade-offs}
\label{subsec:selection_tradeoff}

The multiplier and bootstrap procedures provide two different ways of
approximating the same first-order limiting distribution. Under
Assumption~\ref{ass:copula_first_stage}, the candidate set contains a uniquely
best copula family that is separated from the remaining candidates in the
population criterion. As a result, the AIC weight on the limiting family
converges to one, while the weights on the other families vanish
asymptotically. Model-selection uncertainty is therefore negligible at first
order under the maintained theory. This is why both
Theorem~\ref{thm:multiplier} and
Theorem~\ref{thm:bootstrap_cv_consistency} consistently estimate the same
oracle limiting distribution.

This asymptotic equivalence does not imply that model selection is innocuous
in finite samples. When several candidate copula families fit the data
similarly well, their AIC criteria can be close, so relatively small sampling
perturbations may induce substantial changes in the selected family or in the
associated model-averaging weights. Such variation is not a separate
first-order component under the unique-best-family asymptotics, but it can
nevertheless affect finite-sample calibration.

The two resampling procedures treat this source of uncertainty differently.
The full bootstrap recomputes the empirical margins, candidate-specific copula
parameters, AIC criteria, model weights, and conditional distribution
estimator within every bootstrap replication. It therefore reproduces the
entire estimation-and-selection pipeline and allows model support to vary
endogenously across resamples. By contrast, the influence-function multiplier
procedure starts from a first-order linear approximation around the fitted
sample. The nonlinear weight perturbation in
Remark~\ref{rem:nonlinear} is designed to make this approximation responsive
to finite-sample variation in relative model support, but it does not require
full re-estimation of all candidate copula models in every multiplier draw.

The distinction creates a practical trade-off. The bootstrap more directly
propagates uncertainty generated by the complete estimation and model-selection
procedure, but this comes at a substantially higher computational cost. The
influence-function multiplier procedure is much cheaper because the
candidate-model estimates and the required influence-function components are
computed once and subsequently reused across multiplier draws. Its purpose is
therefore not to replace the full bootstrap uniformly, but to provide a
computationally efficient approximation that retains the same first-order
validity under the maintained assumptions while partially accommodating
finite-sample model-selection variability. Section~\ref{sec:simulation}
examines how these differences translate into finite-sample size, power, and
computational performance.
\section{Simulation Study}
\label{sec:simulation}

This section evaluates the finite-sample calibration, power, and computational
trade-offs of the two inference procedures discussed in
Section~\ref{subsec:selection_tradeoff}. We consider three heterogeneous-\(X\)
designs. The two independent samples have equal size,
\(n_1=n_2=n\in\{50,100,200,500\}\). Every design--sample-size cell is based
on 1,000 Monte Carlo replications at the nominal 5\% level. Both procedures
use a \(25\times25\) empirical-quantile grid, \(\varepsilon=0.02\),
\(a_N=\sqrt{\log(n_1+n_2)}\), and 999
resampling draws.

We compare the AIF multiplier procedure with the
full bootstrap, using the same canonical maximum pseudo-likelihood and
AIC model-averaged point estimator in both cases. As described in
Section~\ref{sec:critical_values}, the full bootstrap recomputes the margins,
candidate copula estimates, AIC weights, and conditional CDFs in every
resampling draw, whereas the AIF procedure uses the nonlinear multiplier
implementation in Remark~\ref{rem:nonlinear} to perturb the fitted
candidate-specific conditional CDFs and their relative model weights without
fully re-estimating all candidate models in each draw.

The candidate set comprises the independence, Gaussian, Student-\(t\),
Clayton, Gumbel, Frank, and Joe copulas, together with the
survival, \(90^\circ\)-rotated, and \(270^\circ\)-rotated versions of the
Clayton and Gumbel copulas (13 families in total). For the Student-\(t\)
copula, the degrees of freedom range over 3, 4, 5, 7, 10, 15, and 30.

\subsection{Data-generating processes}

The tuning parameter \(\delta\) controls a location shift in population 2.
At \(\delta=0\), each design is a boundary-null experiment; positive values
generate \(H_1(y\mid x)-H_2(y\mid x)>0\) on part of the evaluation region.

\begin{itemize}
\item \textbf{P-DGP 1 (correctly specified Gaussian copula).}
Let \(\rho=0.5\), \(X^{(1)}\sim N(0,1)\), and
\(X^{(2)}\sim N(0.5,1)\).  With independent standard normal errors,
\[
Y^{(1)}=\rho X^{(1)}+\sqrt{1-\rho^2}\,e^{(1)},\qquad
Y^{(2)}=\rho X^{(2)}+\delta+\sqrt{1-\rho^2}\,e^{(2)},
\]
where \(\delta\in\{0,0.1,\ldots,0.6\}\).  The two \(X\) margins differ, but
the Gaussian copula is contained in the candidate set.

\item \textbf{P-DGP 2 (correctly specified Clayton copula with heterogeneous
margins).}
Independently for \(g=1,2\), draw \((U_g,V_g)\) from a Clayton copula with
parameter \(\theta=2\).  Define
\[
q(p)=\bigl[|p-1/2|-0.4\bigr]_+^2,
\qquad
L(p)=\frac{\sqrt{3}}{\pi}\log\!\left(\frac{p}{1-p}\right),
\]
and set
\[
\begin{aligned}
X^{(1)}&=\Phi^{-1}(U_1),
&Y^{(1)}&=L(V_1)+10q(V_1),\\
X^{(2)}&=\Phi^{-1}(U_2)+10q(U_2),
&Y^{(2)}&=L(V_2)+\delta .
\end{aligned}
\]
The tail perturbation is zero on the central 80\% probability region.  The
four group-specific marginal maps are all strictly increasing on \((0,1)\):
\(\Phi^{-1}\), \(L+10q\), \(\Phi^{-1}+10q\), and \(L+\delta\).  Indeed,
\(10q'(p)\geq -2\), whereas
\((\Phi^{-1})'(p)\geq\sqrt{2\pi}>2\) and
\(L'(p)\geq 4\sqrt{3}/\pi>2\).  Hence each \((X^{(g)},Y^{(g)})\) is
obtained from its latent Clayton pair through strictly increasing, though
group-specific, marginal transformations, so its copula remains exactly
Clayton with parameter \(\theta=2\).  The perturbation creates markedly
different marginal distributions while preserving the central contact
region.  We take \(\delta\in\{0,0.1,\ldots,0.6\}\).

\item \textbf{P-DGP 3 (misspecified nonlinear stress design).}
Let \(X^{(1)}\sim U(0,1)\), \(X^{(2)}\sim\operatorname{Beta}(2,2)\),
\(\sigma=0.5\), and
\[
m(x)=2(x-0.5)+0.5(x-0.5)^2.
\]
Generate
\[
Y^{(1)}=m(X^{(1)})+\sigma e^{(1)},\qquad
Y^{(2)}=m(X^{(2)})+\delta+\sigma e^{(2)},
\]
for \(\delta\in\{0,0.1,\ldots,0.5\}\).  The conditional distribution is
identical across groups at \(\delta=0\), but the induced copula is generally
outside the finite candidate set.  This design therefore provides robustness
evidence under misspecification and is not covered by the correct-
specification condition in Assumption~\ref{ass:copula_first_stage}.
\end{itemize}

\subsection{Simulation results}

Figures~\ref{fig:power-dgp1}--\ref{fig:power-dgp3} summarize the rejection
probabilities over \(\delta\). The value at \(\delta=0\) reports empirical
size, while positive values of \(\delta\) correspond to increasingly separated
alternatives. The AIF procedure is close to the nominal level in P-DGP 1 and
P-DGP 3 and is conservative in P-DGP 2. The bootstrap shows some
over-rejection at smaller sample sizes in P-DGP 1 and P-DGP 3 and is
conservative in P-DGP 2 at \(n=500\).

Across all three designs, power increases with both \(n\) and \(\delta\),
showing that both procedures become more discriminating as the alternative
moves farther from the boundary null and as sampling information increases.
The AIF and bootstrap procedures exhibit broadly similar power patterns
throughout the simulations. In P-DGP 1, the two power curves are closely
aligned under correct copula specification. A similarly close pattern is
observed in the misspecified P-DGP 3, which provides finite-sample robustness
evidence outside the formal validity conditions of the theory. P-DGP 2 shows
somewhat more conservative rejection behavior for the AIF procedure, but the
overall power pattern remains comparable across the two methods and the
difference narrows as \(n\) and \(\delta\) increase.

\begin{figure}[p]
\centering
\includegraphics[width=\textwidth]{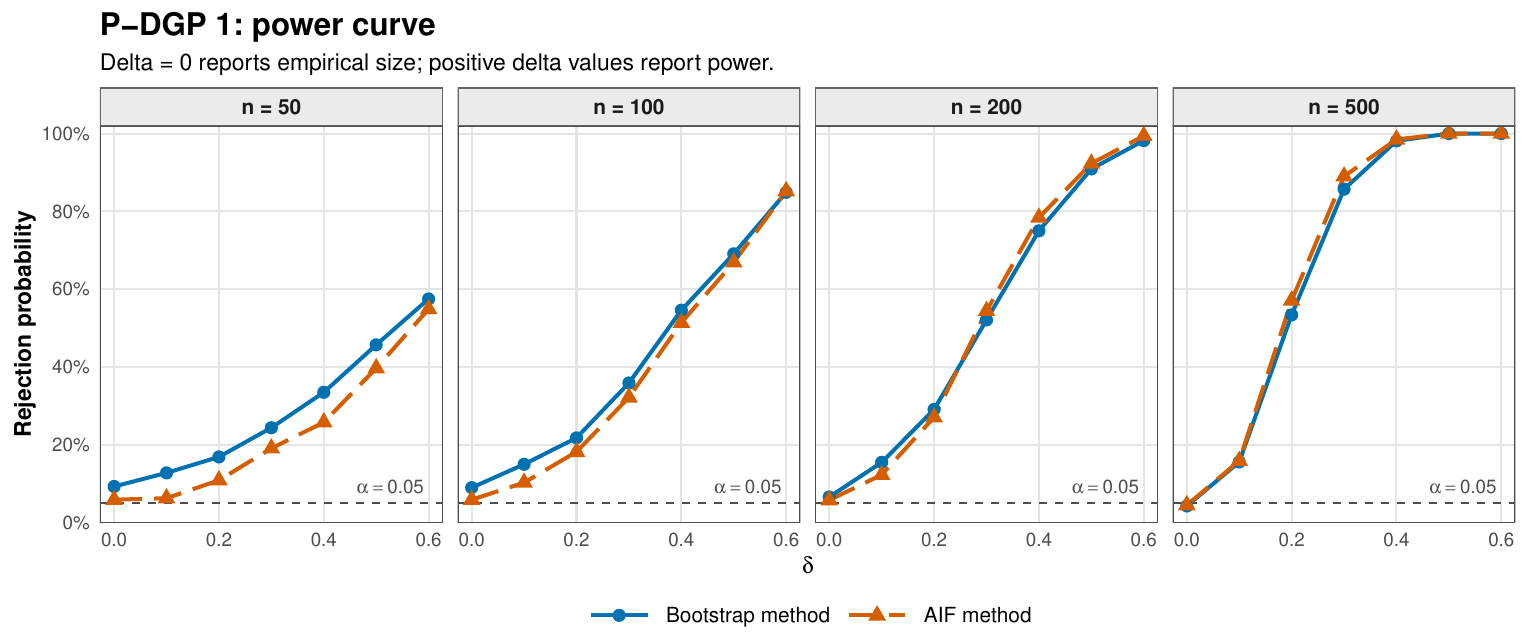}
\caption{Power curves in P-DGP 1.}
\label{fig:power-dgp1}
\end{figure}

\begin{figure}[p]
\centering
\includegraphics[width=\textwidth]{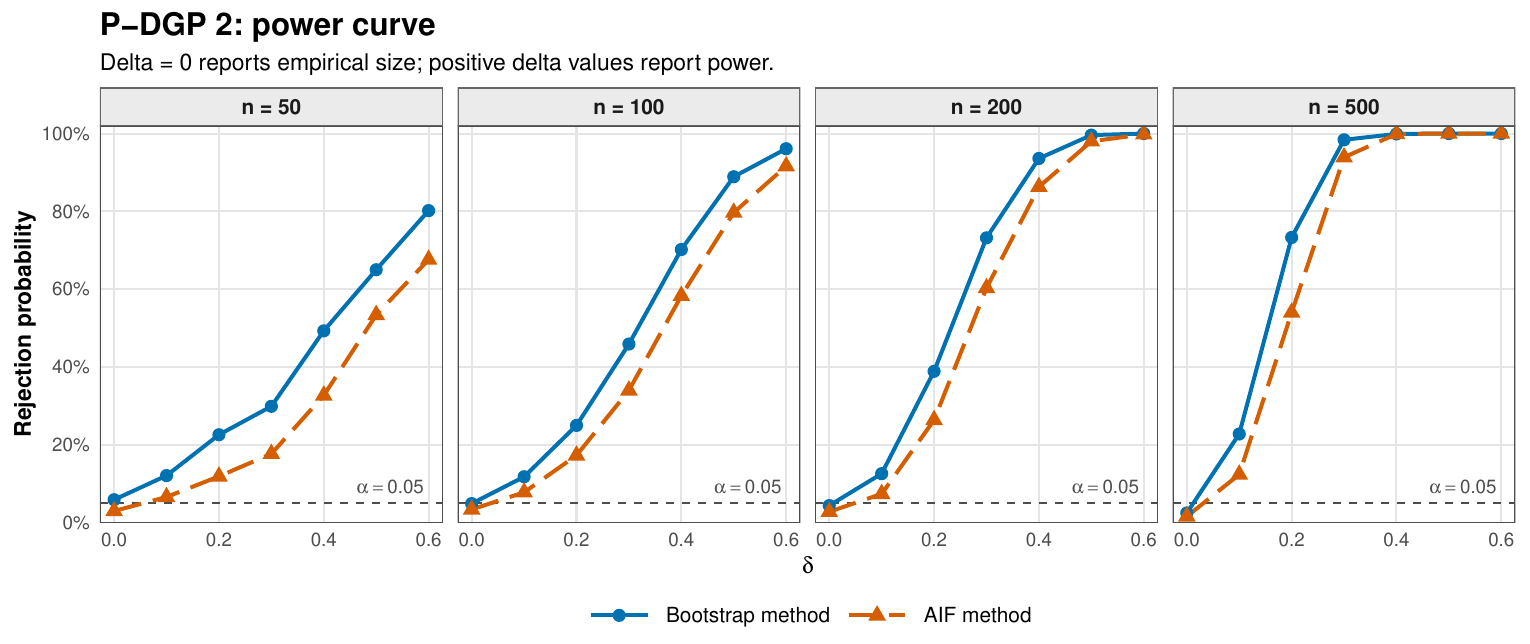}
\caption{Power curves in P-DGP 2.}
\label{fig:power-dgp2}
\end{figure}

\begin{figure}[p]
\centering
\includegraphics[width=\textwidth]{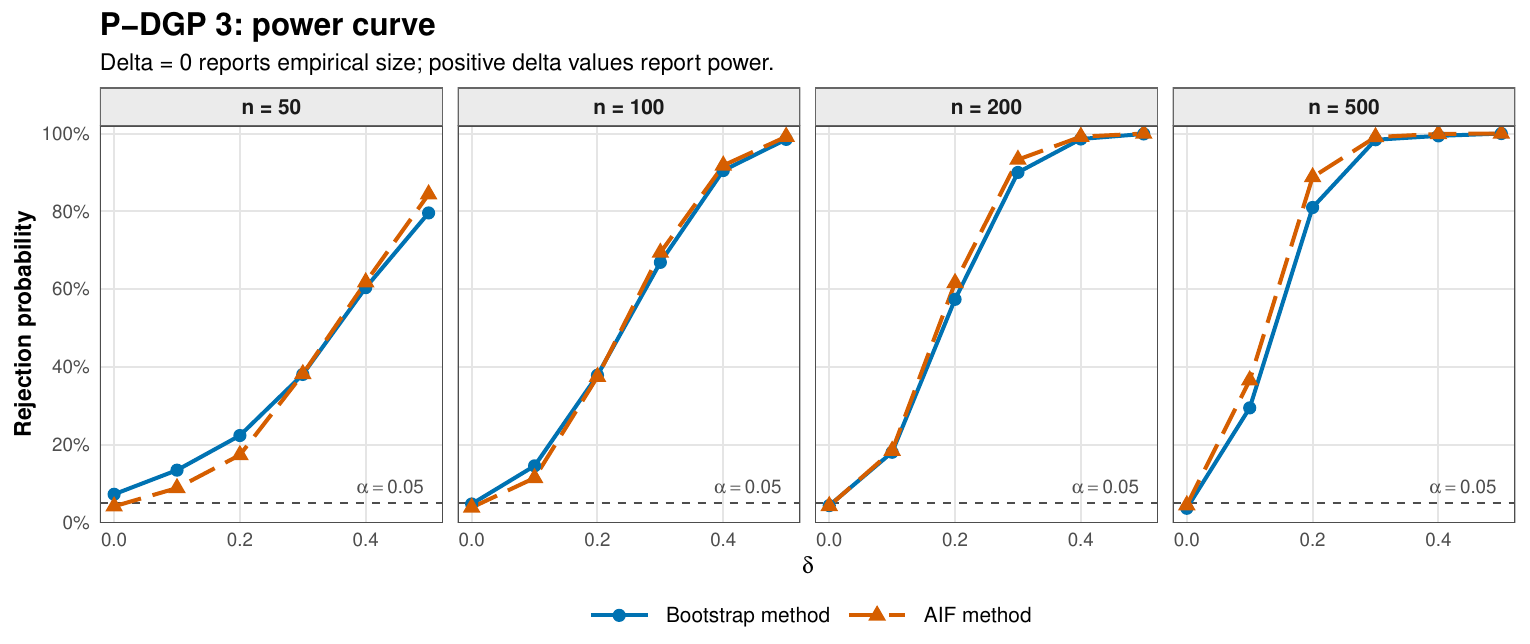}
\caption{Power curves in the misspecified P-DGP 3.}
\label{fig:power-dgp3}
\end{figure}

Table~\ref{tab:dgp1-selection-size} examines size in P-DGP 1 and includes an
additional boundary-null experiment at \(n=1000\), using the same 1,000 Monte
Carlo replications and 999 resampling draws. The bootstrap method's
small-sample over-rejection diminishes as \(n\) increases, with the rejection
rate falling from 0.093 at \(n=50\) to 0.043 at \(n=500\) and 0.042 at
\(n=1000\). The AIF procedure remains closer to the nominal 5\% level at the
smaller sample sizes, with rejection rates between 0.045 and 0.059 across the
five sample sizes. Its rejection rate is not monotone in \(n\), rising from
0.045 at \(n=500\) to 0.057 at \(n=1000\), although this difference is modest
relative to the Monte Carlo standard error of approximately 0.007 for a
rejection probability near 0.05 based on 1,000 replications.

The selection strata help explain the finite-sample calibration pattern. The
proportion of replications in which both populations select the correctly
specified Gaussian family rises from 7.2\% at \(n=50\) to 69.8\% at \(n=500\)
and 83.4\% at \(n=1000\), indicating that model selection becomes progressively
more stable with sample size. For \(n\le500\), bootstrap rejection rates are
higher when at least one population selects a non-Gaussian family, while the
AIF procedure has lower rejection rates within this stratum. This ordering
disappears at \(n=1000\), when selection is substantially more stable. These
results are consistent with model-selection uncertainty being an important
source of finite-sample size distortion: when candidate copulas are not sharply
separated, sampling variation can alter the selected family or the associated
AIC weights, whereas this source of error diminishes as selection stabilizes.
The nonlinear AIF implementation is designed to mitigate this effect by
allowing relative model support to vary across multiplier draws and thereby
partially propagating model-selection uncertainty. This pattern is also
consistent with the maintained unique-best-family asymptotics, under which
selection uncertainty vanishes and the AIF and bootstrap procedures share the
same first-order limit.

\begin{table}[htbp]
\centering
\caption{P-DGP 1 empirical size by copula-selection event}
\label{tab:dgp1-selection-size}

\begin{threeparttable}
\footnotesize
\setlength{\tabcolsep}{3.2pt}
\renewcommand{\arraystretch}{1.10}

\newcommand{\sizecell}[2]{%
    \makecell[c]{#1\\[-1pt]{\scriptsize(#2)}}%
}

\begin{tabular}{@{}llccccc@{}}
\toprule
& &
\multicolumn{5}{c}{Sample size per group} \\
\cmidrule(lr){3-7}
Method
& Selection event
& \(n=50\)
& \(n=100\)
& \(n=200\)
& \(n=500\)
& \(n=1000\) \\
\midrule

\multirow{3}{*}{Bootstrap}
& Overall
& \sizecell{0.093}{1000}
& \sizecell{0.090}{1000}
& \sizecell{0.066}{1000}
& \sizecell{0.043}{1000}
& \sizecell{0.042}{1000} \\

& Both Gaussian
& \sizecell{0.056}{72}
& \sizecell{0.034}{178}
& \sizecell{0.035}{423}
& \sizecell{0.026}{698}
& \sizecell{0.044}{834} \\

& \makecell[l]{At least one\\non-Gaussian}
& \sizecell{0.096}{928}
& \sizecell{0.102}{822}
& \sizecell{0.088}{577}
& \sizecell{0.083}{302}
& \sizecell{0.030}{166} \\

\addlinespace[0.35em]

\multirow{3}{*}{AIF}
& Overall
& \sizecell{0.059}{1000}
& \sizecell{0.059}{1000}
& \sizecell{0.058}{1000}
& \sizecell{0.045}{1000}
& \sizecell{0.057}{1000} \\

& Both Gaussian
& \sizecell{0.042}{72}
& \sizecell{0.039}{178}
& \sizecell{0.043}{423}
& \sizecell{0.033}{698}
& \sizecell{0.061}{834} \\

& \makecell[l]{At least one\\non-Gaussian}
& \sizecell{0.060}{928}
& \sizecell{0.063}{822}
& \sizecell{0.069}{577}
& \sizecell{0.073}{302}
& \sizecell{0.036}{166} \\

\bottomrule
\end{tabular}

\begin{tablenotes}[flushleft]
\scriptsize
\item \textit{Notes:}
Entries report rejection rates at the nominal 5\% level.
The number of Monte Carlo replications in each selection stratum is reported
in parentheses. ``Overall'' denotes the unconditional rejection rate across
all 1,000 replications.
\end{tablenotes}

\end{threeparttable}
\end{table}
Finally, the AIF method takes approximately 1--5 seconds per replication,
compared with 159--898 seconds for the bootstrap method.  Figure~
\ref{fig:simulation-runtime} summarizes the difference across designs and
sample sizes.

\begin{figure}[htbp]
\centering
\includegraphics[width=\textwidth]{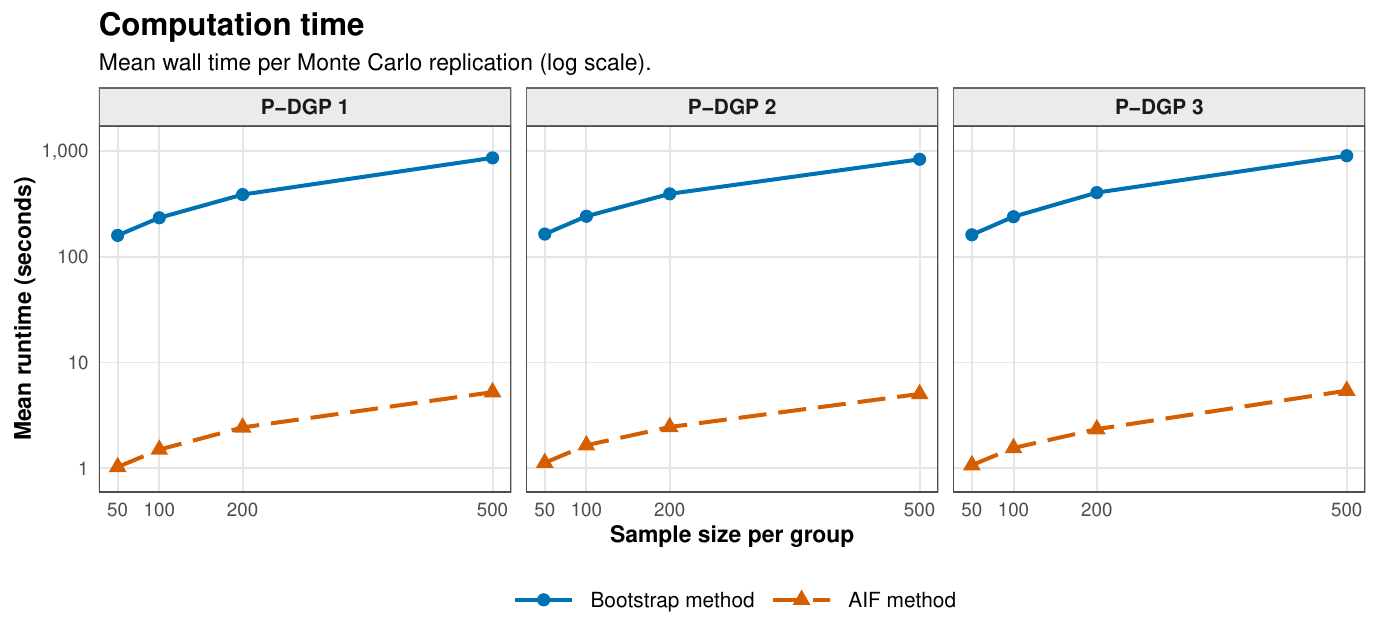}
\caption{Mean computation time per Monte Carlo replication.  The vertical
axis is logarithmic.}
\label{fig:simulation-runtime}
\end{figure}
\section{Empirical Application}
\label{sec:empirical}

We apply the test to intergenerational income mobility using the Panel Study of
Income Dynamics (PSID).  The application asks whether the conditional
distribution of adult income rank differs across parental-education groups at
the same childhood parental-income rank.  It is well suited to the method
because the marginal distribution of childhood income rank differs sharply
across education groups, while the estimand conditions on a common physical
rank value rather than requiring these margins to be equal.

\subsection{Sample and variables}

For each child, \(X\) is the average of the child's within-year parental-family-
income ranks observed at ages 13--17.  The outcome \(Y\) is the average of the
child's within-year adult-family-income ranks observed at ages 30--36.  We
require at least three valid annual observations in each age window and retain
the oldest eligible child from each origin family.  Parental education is the
maximum completed schooling observed for the linked birth parents.  We define
the low, middle, and high groups as at most 12 years, 13--15 years, and at
least 16 years, respectively.  The main analysis is unweighted.

The final sample contains 1,047 individuals: 435 in the low group, 289 in the
middle group, and 323 in the high group.  Table~\ref{tab:empirical-sample}
documents the pronounced differences in the \(X\) margins.  Mean childhood
parental-income rank rises from 0.356 to 0.702 between the low and high
education groups; mean adult-income rank rises from 0.384 to 0.625.  The
empirical CDFs in Figure~\ref{fig:empirical-x-margins} reinforce why a method
that permits population-specific conditioning-variable margins is useful.

\begin{table}[htbp]
\centering
\caption{PSID parental-education sample}
\label{tab:empirical-sample}
\footnotesize
\begin{tabular}{lrrrrr}
\toprule
Group & \(N\) & Mean \(X\) & Median \(X\) & Mean \(Y\) & Median \(Y\) \\
\midrule
Low (\(\leq12\) years) & 435 & 0.356 & 0.324 & 0.384 & 0.356 \\
Middle (13--15 years) & 289 & 0.495 & 0.490 & 0.486 & 0.486 \\
High (\(\geq16\) years) & 323 & 0.702 & 0.749 & 0.625 & 0.669 \\
\bottomrule
\end{tabular}
\end{table}

\begin{figure}[htbp]
\centering
\includegraphics[width=0.94\textwidth]{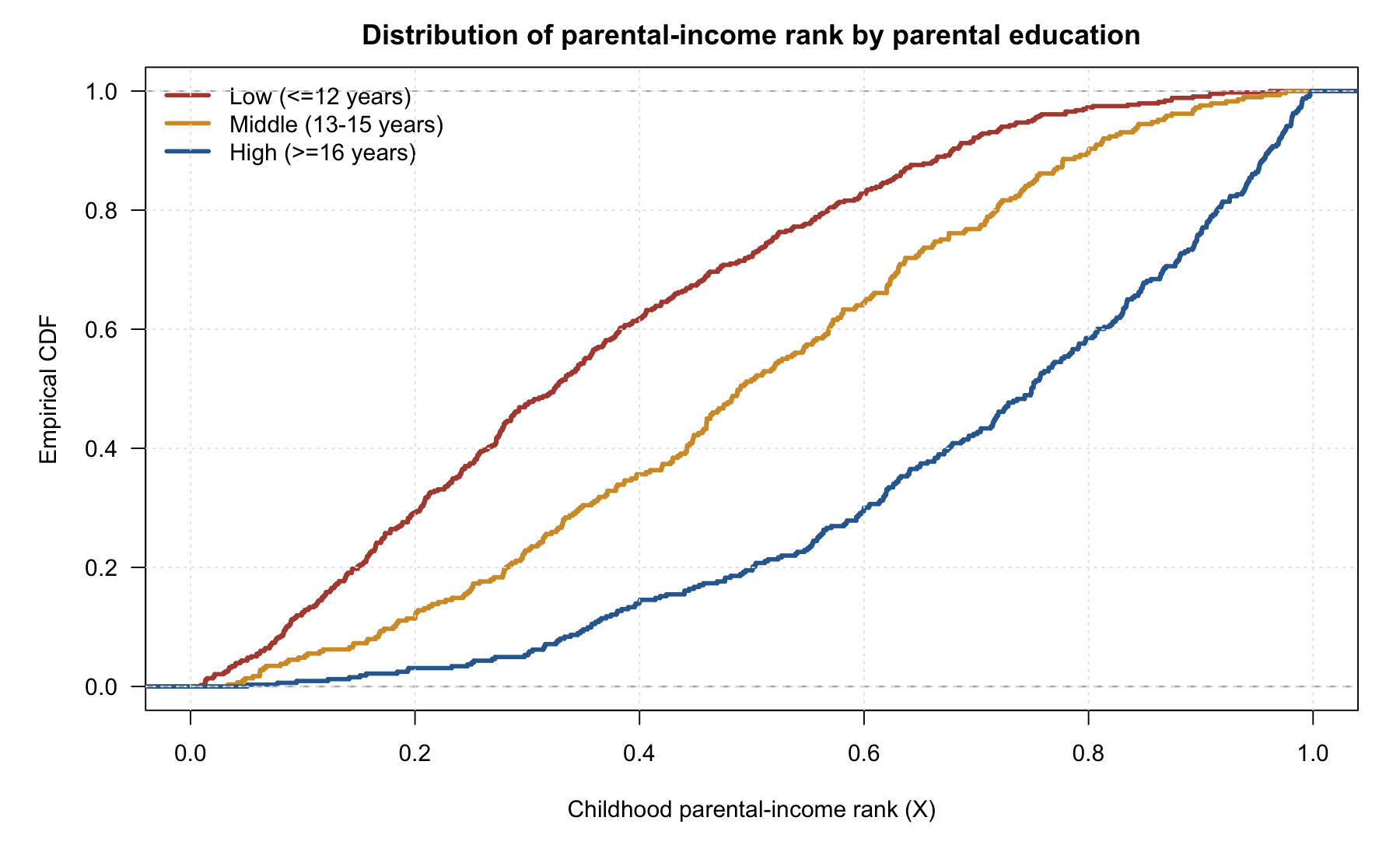}
\caption{Empirical distribution of childhood parental-income rank by parental-
education group.}
\label{fig:empirical-x-margins}
\end{figure}

As a conditional-mean benchmark, an OLS regression of \(Y\) on \(X\), group
indicators, and their interactions yields high- and middle-group indicator
\(p\)-values of 0.357 and 0.729 and interaction \(p\)-values of 0.190 and
0.325.  The distributional comparison below therefore contains information
not summarized by this linear conditional-mean specification.

\subsection{Testing strategy and main results}

For every unordered pair, we test both directions.  The predicted ordering is
that the higher-education group has a better adult-income distribution,
\(H_{\mathrm{higher}}(y\mid x)\leq H_{\mathrm{lower}}(y\mid x)\); the reverse
test interchanges the groups.  We describe an ordering as supported only when the predicted null is not rejected, and the reverse null is rejected.  This two-direction rule does not turn a failure to reject into proof of dominance,
but it rules out the opposite ordering.  We adjust the six directional
\(p\)-values by the Holm procedure.

The baseline AIC model-average specification is based on pseudo-MLE,
\(\varepsilon=0.03\), a \(25\times25\) initial grid, and 1,999 multiplier
draws.  For each unordered pair \((g,h)\), we set
\(a_N=\sqrt{\log(n_g+n_h)}\), as in the simulations.  Let
\[
p_j=\varepsilon+\frac{j-1}{J-1}(1-2\varepsilon),
\qquad j=1,\ldots,J,\qquad J=25.
\]
For each pair, we construct a common grid as the Cartesian product of the
pooled-sample empirical \(p_j\)-quantiles of \(X\) and \(Y\), retaining only
points at which all four group-specific empirical margins lie in
\([\varepsilon,1-\varepsilon]\).  Thus, the predicted and reverse tests for a given pair are evaluated at the
same \((x,y)\) values, although the evaluation grid may differ across pairs.

As a robustness check, we also use 1,999 replications of the full weighted
bootstrap.  In each replication, the empirical margins are re-estimated, all
candidate copula families are refitted for each group, the AIC weights are
recomputed, and the conditional CDFs are reconstructed on the pair-specific
grid.  The full bootstrap therefore propagates uncertainty from marginal
estimation, copula-parameter estimation, and model weighting through the
entire re-estimation procedure, rather than conditioning on the copula
families selected in the original sample.

Table~\ref{tab:empirical-tests} shows that only the high--low comparison
satisfies the two-direction criterion after multiplicity adjustment.  The
predicted high-versus-low null is not rejected by either inference procedure
(raw \(p=1.0000\) for both).  The reverse null is rejected by the
model-average multiplier procedure (raw \(p=0.0010\), Holm-adjusted
\(p=0.0060\)) and by the full  weighted bootstrap (raw
\(p=0.0045\), Holm-adjusted \(p=0.0270\)).  The high--middle and middle--low
comparisons remain inconclusive because neither direction is rejected after
Holm adjustment.  The evidence therefore supports an endpoint separation
between high and low parental education, but not a complete stepwise
education gradient.

\begin{table}[htbp]
\centering
\caption{Directional conditional-dominance tests in the PSID application}
\label{tab:empirical-tests}
\footnotesize
\begin{tabular}{llcccc}
\toprule
Pair & Direction & \multicolumn{2}{c}{AIC-average multiplier} &
\multicolumn{2}{c}{Full AIC-average bootstrap} \\
\cmidrule(lr){3-4}\cmidrule(lr){5-6}
& & Raw \(p\) & Holm \(p\) & Raw \(p\) & Holm \(p\) \\
\midrule
High--Middle & Predicted & 0.9925 & 1.0000 & 0.9945 & 1.0000 \\
             & Reverse   & 0.2795 & 1.0000 & 0.3605 & 1.0000 \\
Middle--Low  & Predicted & 0.9825 & 1.0000 & 0.9875 & 1.0000 \\
             & Reverse   & 0.0410 & 0.2050 & 0.0965 & 0.4825 \\
High--Low    & Predicted & 1.0000 & 1.0000 & 1.0000 & 1.0000 \\
             & Reverse   & 0.0010 & 0.0060 & 0.0045 & 0.0270 \\
\bottomrule
\end{tabular}
\begin{minipage}{0.94\textwidth}
\footnotesize\emph{Notes:} ``Predicted'' tests whether the higher parental-
education group dominates the lower group; ``Reverse'' tests the opposite
ordering.  Holm adjustment is over all six directions within each inference
procedure.  In every weighted-bootstrap draw, the empirical margins and all
13 candidate copulas are re-estimated and the AIC weights are recomputed.
Each inference procedure uses 1,999 draws.
\end{minipage}
\end{table}
The estimated AIC-averaged conditional-CDF difference for the high--low pair
ranges from \(-0.198\) to \(-0.033\) on its pair-specific common grid, has mean
\(-0.112\), and is nonpositive at every evaluated point.  Thus, conditional on
the same observed parental-income rank, the high-education group's estimated
adult-income distribution is uniformly shifted to the right.  The largest
gaps at parental-income ranks near 0.24, 0.50, and 0.76 are 0.137, 0.149, and
0.198 in CDF units.  Figure~\ref{fig:empirical-heatmaps} displays all three
pairwise contrasts; blue cells are consistent with the predicted ordering.
The heatmaps show effect direction and magnitude, not pointwise
\(p\)-values---inference is based on the supremum statistic in
Table~\ref{tab:empirical-tests}.

\begin{figure}[p]
\centering
\includegraphics[width=\textwidth]{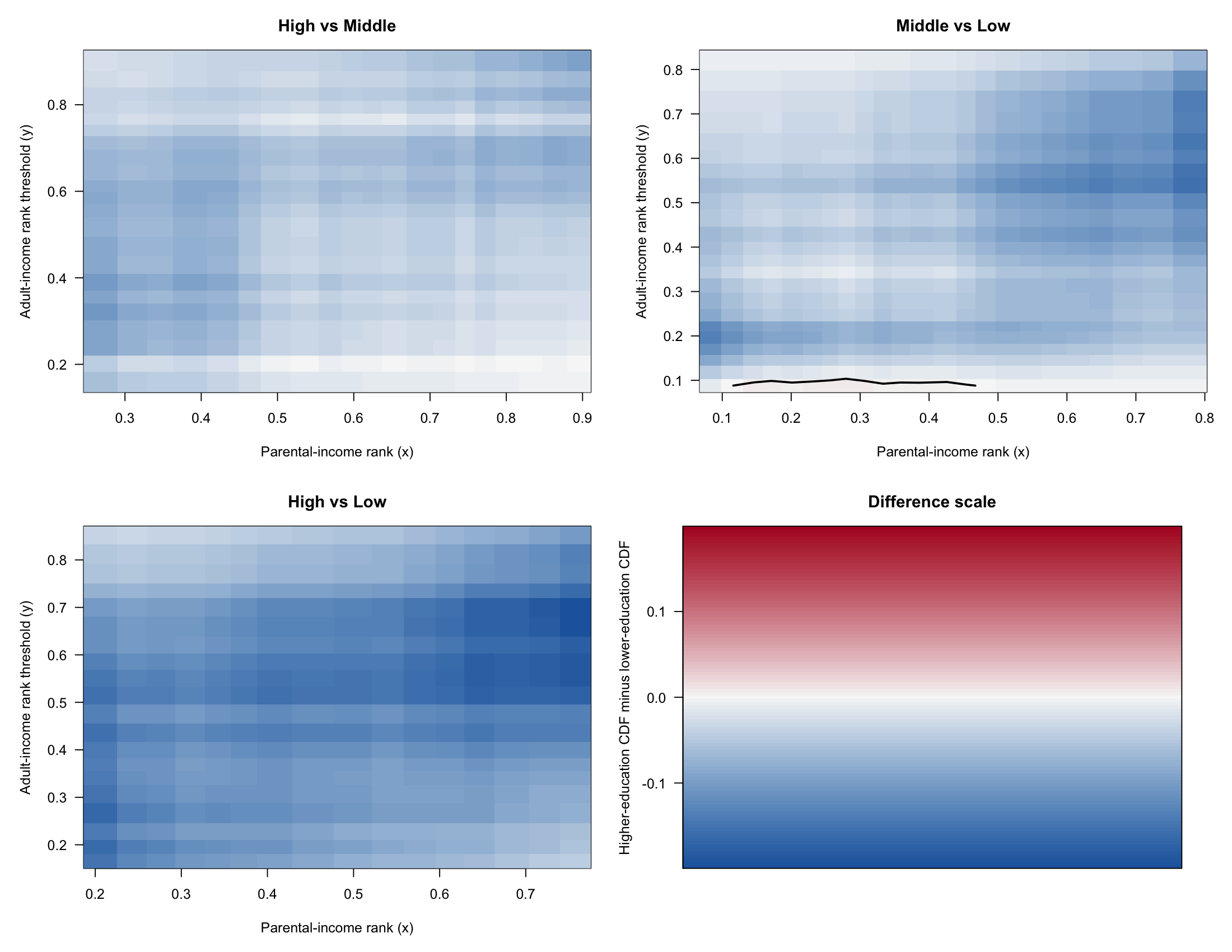}
\caption{AIC-averaged conditional-CDF differences: higher parental education
minus lower parental education.  Negative values (blue) are consistent with
the predicted dominance ordering.  The black curve marks estimated zero
crossings where present.}
\label{fig:empirical-heatmaps}
\end{figure}
\subsection{Robustness and interpretation}

The high--low conclusion remains unchanged when inference is based on the full
AIC-averaged weighted bootstrap.  With 1,999 draws, the reverse null
has a raw \(p\)-value of 0.0045 and a Holm-adjusted \(p\)-value of 0.0270.
Multiplier-based sensitivity checks yield the same qualitative conclusion.
The Holm-adjusted \(p\)-value for the reverse test is 0.006 at trimming values
of 0.05 and 0.10; it is 0.009, 0.015, and 0.003 on 20-, 30-, and 40-point
grids, respectively, and 0.003 under the AIC single-best-family specification.

The full-sample AIC criterion selects Gaussian copulas for the low and middle
groups and a Frank copula for the high group.  As a supplementary diagnostic
for these selected families, we apply White's information-matrix test to each
fitted copula pseudo-likelihood.  The null hypothesis equates the negative
expected Hessian of the copula log-density with the expected outer product of
its score, as implied by correct specification of the selected copula.  The
resulting \(p\)-values are 0.930, 0.120, and 0.970 for the low, middle, and high
groups, respectively, and the diagnostic therefore does not reject any of the
selected families.  The test is applied to the full-sample AIC winner in each
group rather than to the AIC-averaged estimator.  In addition, its
implementation treats the estimated margins as known and does not account for
the preceding model-selection step.  We therefore view nonrejection only as
supplementary evidence against substantial misspecification, rather than as
evidence that the selected family is correctly specified or uniquely optimal.

The endpoint conclusion is more sensitive to changes in the target sample and
outcome definition.  In the SRC-only subsample, the high--low reverse test has
a raw \(p\)-value of 0.013 and a Holm-adjusted \(p\)-value of 0.078.  Replacing
adult family-income rank with individual labor-income rank yields corresponding
\(p\)-values of 0.024 and 0.144.  The main finding should therefore be
interpreted as specific to the unweighted family-income analysis.  It reflects
a descriptive conditional-distribution comparison rather than a causal effect
of parental education, which may also proxy for race, location, family
resources, and long-run selection.  Population-representative inference would
further require survey weights to be incorporated consistently into the
empirical margins, copula estimation, and resampling procedure.
\section{Conclusion}

This paper makes a common-value, two-population conditional-distribution surface an operational object of inference when covariate margins differ across populations. The resulting comparison asks whether the entire outcome distributions are ordered along an empirically relevant continuum of physical covariate and outcome values. Its value is therefore research-level rather than estimator-specific: it permits a coherent region-wide comparison while preserving the distinction between a common physical value and a common percentile position. The accompanying inferential contribution turns the fitted surface into simultaneous evidence for a uniform ordering, rather than leaving it as a visualization or a collection of unadjusted pointwise contrasts.

The copula-derivative representation supplies the structure needed to implement this comparison. Population-specific margins locate the same physical covariate value within each group, and the fitted dependence model links conditional distributions across the region. Uniform process and contact-set arguments then propagate uncertainty from the estimated margins and dependence parameters through a one-sided statistic whose binding locations are unknown. The formal guarantees require correct specification within a finite copula class, smoothness and trimming conditions, and a uniquely best candidate family. Thus, the method offers an explicit structure--flexibility tradeoff rather than a claim that dependence modeling is costless.

The simulations reinforce both the value and the limits of the inferential
guarantee.  Rejection probabilities generally increase as alternatives become
easier to distinguish, but small-sample calibration remains sensitive to
model selection and varies across designs.  Evidence from the misspecified
design is informative about robustness without extending the theory beyond
its stated regime.  In the descriptive PSID application, both inference
procedures support the high--low endpoint comparison.  The predicted
high-versus-low null is not rejected, whereas the reverse null is rejected
after multiplicity adjustment by both the multiplier procedure and the full
 weighted bootstrap; for the latter, the Holm-adjusted \(p\)-value
is 0.0270.  The adjacent high--middle and middle--low comparisons remain
inconclusive.  The empirical evidence therefore supports bounded endpoint
separation rather than a complete stepwise education gradient.  This finding
is specific to the unweighted family-income analysis and does not establish a
causal effect of parental education.  Future work can address model ties,
broader misspecification, and weighted population comparisons.  More
generally, the paper shows how an explicit dependence structure can make
region-wide distributional comparisons feasible while keeping their
inferential and empirical boundaries visible.
\clearpage

\bibliographystyle{apalike}
\bibliography{sample}

\newpage
\appendix
\section{Proofs}
\label{app:proofs}

\subsection*{Proof of Lemma~\ref{lem:oracle_equivalence}}
\begin{proof}
Fix \(g\in\{1,2\}\). For each candidate model, write
\[
    \widehat h_{g\ell}(x,y)
    =
    \partial_u C_\ell
    \{\widehat F_{gX}(x),\widehat F_{gY}(y);\widehat\theta_{g\ell}\}.
\]
The model-averaged estimator is
\[
    \widehat H_g(y\mid x)
    =
    \sum_{\ell=1}^Q
    \widehat w_{g\ell}\widehat h_{g\ell}(x,y).
\]
Let
\[
    \widehat H_g^{\,o}(y\mid x)
    =
    \widehat h_{g\ell_g^\ast}(x,y)
    =
    \partial_u C_{\ell_g^\ast}
    \{\widehat F_{gX}(x),\widehat F_{gY}(y);\widehat\theta_g\}
\]
denote the oracle estimator based on the uniquely identified limiting copula
model.

By Assumption~\ref{ass:copula_first_stage} and the uniform law of
large numbers for the pseudo-likelihood criteria, for every
\(j\ne\ell_g^\ast\),
\[
    \frac{1}{n_g}
    \left\{
    \operatorname{AIC}_{gj}
    -
    \operatorname{AIC}_{g\ell_g^\ast}
    \right\}
    \overset{p}{\rightarrow} 
    2\kappa_{gj},
    \qquad
    \kappa_{gj}>0.
\]
Since the number of candidate models is fixed, let
\(\kappa_g=\min_{j\ne\ell_g^\ast}\kappa_{gj}>0\). For any
\(c\in(0,\kappa_g)\), with probability approaching one,
\[
    \operatorname{AIC}_{gj}
    -
    \operatorname{AIC}_{g\ell_g^\ast}
    \ge 2c n_g,
    \qquad
    j\ne\ell_g^\ast .
\]
Using the AIC-weight formula, after subtracting
\(\operatorname{AIC}_{g\ell_g^\ast}\) from all criteria,
\[
    \widehat w_{gj}
    =
    \frac{
    \exp[-\{\operatorname{AIC}_{gj}
    -
    \operatorname{AIC}_{g\ell_g^\ast}\}/2]
    }{
    1+
    \sum_{m\ne\ell_g^\ast}
    \exp[-\{\operatorname{AIC}_{gm}
    -
    \operatorname{AIC}_{g\ell_g^\ast}\}/2]
    }.
\]
Hence, uniformly over \(j\ne\ell_g^\ast\),
\[
    \widehat w_{gj}
    =
    O_p\{\exp(-c n_g)\},
    \qquad
    1-\widehat w_{g\ell_g^\ast}
    =
    \sum_{j\ne\ell_g^\ast}\widehat w_{gj}
    =
    O_p\{\exp(-c n_g)\}.
\]

For every candidate copula, \(\partial_u C_\ell(u,v;\theta)\) is a
conditional distribution function in \(v\) and therefore takes values in
\([0,1]\). Hence,
\[
    0
    \le
    \widehat h_{g\ell}(x,y)
    \le
    1,
    \qquad \ell=1,\ldots,Q,
\]
and consequently
\[
    \max_{1\le\ell\le Q}
    \sup_{(x,y)\in\Lambda_\varepsilon}
    |\widehat h_{g\ell}(x,y)|
    \le 1.
\]
Therefore,
\[
\begin{aligned}
&\sup_{(x,y)\in\Lambda_\varepsilon}
\left|
\widehat H_g(y\mid x)-\widehat H_g^{\,o}(y\mid x)
\right|
\\
&\quad
=
\sup_{(x,y)\in\Lambda_\varepsilon}
\left|
\sum_{j\ne\ell_g^\ast}
\widehat w_{gj}
\{\widehat h_{gj}(x,y)-\widehat h_{g\ell_g^\ast}(x,y)\}
\right|
\\
&\quad
\le
\sum_{j\ne\ell_g^\ast}
\widehat w_{gj}
=
O_p\{\exp(-c n_g)\}.
\end{aligned}
\]
Since \(\exp(-c n_g)=o(n_g^{-1/2})\), the desired result follows.
\end{proof}
\subsection*{Proof of Theorem~\ref{thm:linear_h}}
\begin{proof}
Fix \(g\in\{1,2\}\). By Lemma~\ref{lem:oracle_equivalence}, it is enough to
derive the expansion for the oracle estimator
\[
\widehat H_g^{\,o}(y\mid x)
=
\partial_u C_{\ell_g^\ast}
\left(
\widehat F_{gX}(x),
\widehat F_{gY}(y);
\widehat\theta_g
\right),
\]
because
\[
\sup_{z\in\Lambda_\varepsilon}
|\widehat H_g(z)-\widehat H_g^{\,o}(z)|
=
o_p(n_g^{-1/2}).
\]

Let
\[
    m_g(u,v,\theta)
    =
    \partial_u C_{\ell_g^\ast}(u,v;\theta).
\]
Then
\[
    H_g(y\mid x)
    =
    m_g\{F_{gX}(x),F_{gY}(y),\theta_g\}.
\]
On the trimmed domain, with probability approaching one,
\(\widehat F_{gX}(x)\) and \(\widehat F_{gY}(y)\) remain in a compact subset of
\((0,1)\), uniformly over \((x,y)\in\Lambda_\varepsilon\). By the smoothness
condition in Assumption~\ref{ass:basic_regular}, \(m_g\) is continuously
differentiable in a neighborhood of the true arguments, with derivatives
\[
    \partial_u m_g=C_{g,uu},
    \qquad
    \partial_v m_g=C_{g,uv},
    \qquad
    \partial_\theta m_g=C_{g,u\theta}.
\]

Define
\[
\delta_{gn}
=
\sup_x|\widehat F_{gX}(x)-F_{gX}(x)|
+
\sup_y|\widehat F_{gY}(y)-F_{gY}(y)|
+
\|\widehat\theta_g-\theta_g\|.
\]
The empirical marginal processes are root-\(n_g\) bounded, and
Assumption~\ref{ass:copula_first_stage} gives
\(\widehat\theta_g-\theta_g=O_p(n_g^{-1/2})\). Hence
\[
    \delta_{gn}=O_p(n_g^{-1/2}).
\]
A uniform first-order Taylor expansion gives
\[
\begin{aligned}
\widehat H_g^{\,o}(y\mid x)-H_g(y\mid x)
=
&\ C_{g,uu}(x,y)
\{\widehat F_{gX}(x)-F_{gX}(x)\}
\\
&+
C_{g,uv}(x,y)
\{\widehat F_{gY}(y)-F_{gY}(y)\}
\\
&+
C_{g,u\theta}(x,y)^\top
(\widehat\theta_g-\theta_g)
+
r_{gn}(x,y).
\end{aligned}
\]
The remainder satisfies
\[
    \sup_{z\in\Lambda_\varepsilon}|r_{gn}(z)|
    \le
    \omega_g(\delta_{gn})\delta_{gn},
\]
where \(\omega_g(\cdot)\) is a modulus of continuity for the first derivatives
of \(m_g\) on the relevant compact neighborhood. Since
\(\omega_g(\delta_{gn})=o_p(1)\) and
\(\delta_{gn}=O_p(n_g^{-1/2})\),
\[
    \sup_{z\in\Lambda_\varepsilon}|r_{gn}(z)|
    =
    o_p(n_g^{-1/2}).
\]

For the scaled empirical marginal distribution,
\[
    \widehat F_{gX}(x)
    =
    \frac{1}{n_g+1}
    \sum_{i=1}^{n_g}\mathbf 1(X_{gi}\le x),
\]
we have uniformly in \(x\),
\[
\sqrt{n_g}
\{\widehat F_{gX}(x)-F_{gX}(x)\}
=
\frac{1}{\sqrt{n_g}}
\sum_{i=1}^{n_g}
\{\mathbf 1(X_{gi}\le x)-F_{gX}(x)\}
+
o_p(1).
\]
The same expansion holds uniformly for \(\widehat F_{gY}\). Moreover, by
Assumption~\ref{ass:copula_first_stage},
\[
\sqrt{n_g}(\widehat\theta_g-\theta_g)
=
\frac{1}{\sqrt{n_g}}
\sum_{i=1}^{n_g}
\psi_g(Z_{gi})
+
o_p(1).
\]
Substituting these three expansions into the Taylor expansion yields
\[
\begin{aligned}
&\sqrt{n_g}\{\widehat H_g^{\,o}(y\mid x)-H_g(y\mid x)\}
\\
&\quad =
\frac{1}{\sqrt{n_g}}
\sum_{i=1}^{n_g}
\Big[
C_{g,uu}(x,y)
\{\mathbf 1(X_{gi}\le x)-F_{gX}(x)\}
\\
&\qquad\qquad
+
C_{g,uv}(x,y)
\{\mathbf 1(Y_{gi}\le y)-F_{gY}(y)\}
+
C_{g,u\theta}(x,y)^\top\psi_g(Z_{gi})
\Big]
+
o_p(1),
\end{aligned}
\]
uniformly over \((x,y)\in\Lambda_\varepsilon\). Combining this with
Lemma~\ref{lem:oracle_equivalence} proves the uniform asymptotic linear
representation for \(\widehat H_g\).

It remains to verify weak convergence. Consider the class
\[
    \mathcal F_g
    =
    \{\phi_{g,z}:z\in\Lambda_\varepsilon\}.
\]
The first component of \(\phi_{g,z}\) belongs to a class of the form
\[
    \left\{
    a(x,y)\{\mathbf 1(X_g\le x)-F_{gX}(x)\}:
    (x,y)\in\Lambda_\varepsilon
    \right\},
\]
where \(a(x,y)=C_{g,uu}(x,y)\) is uniformly bounded. The centered half-line
indicator class is a VC class and hence Donsker. Multiplication by a uniformly
bounded scalar coefficient preserves the Donsker property. The same argument
applies to the second component involving
\(\mathbf 1(Y_g\le y)-F_{gY}(y)\).

For the third component, the class is
\[
    \left\{
    C_{g,u\theta}(x,y)^\top\psi_g(Z_g):
    (x,y)\in\Lambda_\varepsilon
    \right\}.
\]
Since \(C_{g,u\theta}(x,y)\) ranges over a bounded subset of a finite-dimensional
Euclidean space and \(E\|\psi_g(Z_g)\|^2<\infty\), this is a finite-dimensional
linear class with a square-integrable envelope. Hence it is Donsker. Therefore,
by permanence of the Donsker property under finite sums,
\(\mathcal F_g\) is \(P_g\)-Donsker.

The empirical process central limit theorem then gives
\[
    \frac{1}{\sqrt{n_g}}
    \sum_{i=1}^{n_g}
    \phi_{g,\cdot}(Z_{gi})
    \rightsquigarrow
    \mathbb G_g
    \quad
    \text{in }
    \ell^\infty(\Lambda_\varepsilon),
\]
where \(\mathbb G_g\) is a tight mean-zero Gaussian process with covariance
kernel \(\Sigma_g\). The theorem follows.
\end{proof}

\subsection*{Proof of Theorem~\ref{thm:two_sample_sup_limit}}
\begin{proof}
By Theorem~\ref{thm:linear_h}, for \(g=1,2\),
\[
    \sqrt{n_g}\{\widehat H_g-H_g\}
    \rightsquigarrow
    \mathbb G_g
    \quad
    \text{in }
    \ell^\infty(\Lambda_\varepsilon).
\]
Since the two samples are independent, the two processes converge jointly to
independent Gaussian limits \((\mathbb G_1,\mathbb G_2)\).

Because
\[
    \widehat\Delta-\Delta
    =
    (\widehat H_1-H_1)-(\widehat H_2-H_2),
\]
we have
\[
\begin{aligned}
    s_N(\widehat\Delta-\Delta)
    =
    \sqrt{\frac{n_2}{n_1+n_2}}\,
    \sqrt{n_1}(\widehat H_1-H_1)
    -
    \sqrt{\frac{n_1}{n_1+n_2}}\,
    \sqrt{n_2}(\widehat H_2-H_2).
\end{aligned}
\]
Since \(n_1/(n_1+n_2)\to\lambda\), Slutsky's theorem and the continuous mapping
theorem yield
\[
    s_N(\widehat\Delta-\Delta)
    \rightsquigarrow
    \sqrt{1-\lambda}\,\mathbb G_1
    -
    \sqrt{\lambda}\,\mathbb G_2
    =
    \mathbb G_\Delta
\]
in \(\ell^\infty(\Lambda_\varepsilon)\).

Let \(\mathcal S:\ell^\infty(\Lambda_\varepsilon)\to\mathbb R\) denote the
supremum functional, \(\mathcal S(f)=\sup_{z\in\Lambda_\varepsilon}f(z)\).
Then
\(T_N=s_N\mathcal S(\widehat\Delta),\)
and hence
\[
    T_N-s_N\mathcal S(\Delta)
    =
    s_N\{\mathcal S(\widehat\Delta)-\mathcal S(\Delta)\}.
\]

Since \(\Lambda_\varepsilon\) is compact and \(\Delta\) is continuous,
\(\mathcal S\) is Hadamard directionally differentiable at \(\Delta\),
tangentially to \(C(\Lambda_\varepsilon)\)
\citep[Lemma~S.4.9]{FangSantos2019}. Its directional derivative is
\[
    \mathcal S'_{\Delta}(h)
    =
    \sup_{z\in\Gamma^\ast(\Delta)}h(z),
\]
where
\[
    \Gamma^\ast(\Delta)
    =
    \left\{
    z\in\Lambda_\varepsilon:
    \Delta(z)=\sup_{\tilde z\in\Lambda_\varepsilon}\Delta(\tilde z)
    \right\}
\]
is the set of maximizers of \(\Delta\).

We have
\[
    s_N(\widehat\Delta-\Delta)
    \rightsquigarrow
    \mathbb G_\Delta
    \quad
    \text{in }
    \ell^\infty(\Lambda_\varepsilon).
\]
Because \(\mathbb G_\Delta\) has a version in \(C(\Lambda_\varepsilon)\), the
directional functional delta method of
\citet[Theorem~2.1]{Shapiro1991} and
\citet[Proposition~1]{dumbgen} gives
\[
    s_N\{\mathcal S(\widehat\Delta)-\mathcal S(\Delta)\}
    \rightsquigarrow
    \mathcal S'_{\Delta}(\mathbb G_\Delta)
    =
    \sup_{z\in\Gamma^\ast(\Delta)}
    \mathbb G_\Delta(z).
\]
Equivalently, because \(T_N=s_N\mathcal S(\widehat\Delta)\),
\[
    T_N
    -
    s_N\sup_{z\in\Lambda_\varepsilon}\Delta(z)
    \rightsquigarrow
    \sup_{z\in\Gamma^\ast(\Delta)}
    \mathbb G_\Delta(z).
\]
This completes the proof.
\end{proof}

\subsection*{Proof of Theorem~\ref{thm:oracle_test_validity}}

\begin{proof}
Under the boundary null,
\(\sup_{z\in\Lambda_\varepsilon}\Delta(z)=0\), so
\(\Gamma^\ast(\Delta)=\Gamma(\Delta)\). Therefore,
Theorem~\ref{thm:two_sample_sup_limit} gives
\[
    T_N
    \rightsquigarrow
    \sup_{z\in\Gamma(\Delta)}
    \mathbb G_\Delta(z).
\]
Since \(c_{1-\alpha}\) is the \((1-\alpha)\)-quantile of this
limiting distribution and the distribution is continuous at
\(c_{1-\alpha}\),
\[
    \lim_{n_1,n_2\to\infty}
    \Prob(T_N>c_{1-\alpha})
    \leq \alpha.
\]

If the null is strict, then by compactness of \(\Lambda_\varepsilon\) and
continuity of \(\Delta\), there exists \(\eta>0\) such that
\[
    \sup_{z\in\Lambda_\varepsilon}\Delta(z)\leq -\eta.
\]
Theorem~\ref{thm:two_sample_sup_limit} implies
\[
    T_N
    =
    s_N\sup_{z\in\Lambda_\varepsilon}\Delta(z)
    +
    O_p(1)
    \leq
    -s_N\eta+O_p(1)
    \overset{p}{\longrightarrow}
    -\infty.
\]
Since \(c_{1-\alpha}=0\) under the strict null,
\[
    \Prob(T_N>c_{1-\alpha})\to0.
\]
Together, the boundary and strict-null cases establish the stated
asymptotic size control.

If a fixed alternative holds, then
\(\sup_{z\in\Lambda_\varepsilon}\Delta(z)=\eta>0\) for some \(\eta\). Hence
\[
    T_N
    =
    s_N\eta+O_p(1)
    \overset{p}{\longrightarrow}
    +\infty.
\]
Since the oracle critical value is finite,
\[
    \Prob(T_N>c_{1-\alpha})\to1,
\]
which establishes consistency.
\end{proof}

\subsection*{Proof of Theorem~\ref{thm:multiplier}}

For any \(\eta>0\), define the \(\eta\)-neighborhood of the contact set by
\[
    \Gamma^\eta(\Delta)
    =
    \left\{
    z\in\Lambda_\varepsilon:
    d\left(z,\Gamma(\Delta)\right)\le \eta
    \right\}.
\]

\begin{lemma}[Consistency of the estimated contact set]
\label{lem:contact_set_consistency}
Under $H_0$ and the same conditions as in Theorem~\ref{thm:two_sample_sup_limit}, suppose that
\(\Gamma(\Delta)\neq\varnothing\). If \(a_N\to\infty\) and
\(a_N/s_N\to0\), then  for every \(\eta>0\),
\[
    \Prob\left\{
    \Gamma(\Delta)\subseteq\widehat\Gamma_N
    \subseteq\Gamma^\eta(\Delta)
    \right\}\to1.
\]
\end{lemma}

\begin{proof}
By Theorem~\ref{thm:two_sample_sup_limit},
\[
    \sup_{z\in\Lambda_\varepsilon}
    \left|
    s_N\{\widehat\Delta(z)-\Delta(z)\}
    \right|
    =
    O_p(1).
\]

First, for any \(z\in\Gamma(\Delta)\), we have \(\Delta(z)=0\). Hence,
uniformly over \(\Gamma(\Delta)\),
\[
    s_N\widehat\Delta(z)
    =
    s_N\{\widehat\Delta(z)-\Delta(z)\}
    =
    O_p(1).
\]
Since \(a_N\to\infty\), it follows that
\[
    \Prob\left(
    \Gamma(\Delta)\subseteq \widehat\Gamma_N
    \right)\to 1.
\]

Next, fix any \(\eta>0\) and define
\[
    A_\eta
    =
    \Lambda_\varepsilon\setminus \Gamma^{\eta}(\Delta).
\]
Because \(\Lambda_\varepsilon\) is compact and \(\Delta\) is continuous,
\(A_\eta\) is compact. Moreover, under the boundary null,
\[
    \sup_{z\in\Lambda_\varepsilon}\Delta(z)=0,
\]
and by the definition of the contact set, \(\Delta(z)<0\) for all
\(z\in A_\eta\). Therefore, there exists \(c_\eta>0\) such that
\[
    \sup_{z\in A_\eta}\Delta(z)\le -c_\eta .
\]
Thus, uniformly over \(z\in A_\eta\),
\[
    s_N\widehat\Delta(z)
    =
    s_N\{\widehat\Delta(z)-\Delta(z)\}
    +
    s_N\Delta(z)
    \le
    O_p(1)-s_N c_\eta .
\]
Since \(a_N/s_N\to0\), we have
\[
    \Prob\left(
    \widehat\Gamma_N\cap A_\eta=\varnothing
    \right)\to1.
\]
Equivalently,
\[
    \Prob\left(
    \widehat\Gamma_N
    \subseteq
    \Gamma^{\eta}(\Delta)
    \right)\to1.
\]
This completes the proof.
\end{proof}

\begin{remark}
\label{rem:strict_null_contact_set}
If \(\Gamma(\Delta)=\varnothing\), then compactness of
\(\Lambda_\varepsilon\) and continuity of \(\Delta\) imply
\[
    \sup_{z\in\Lambda_\varepsilon}\Delta(z)<0,
\]
and the same argument gives
\[
    \Prob(\widehat\Gamma_N=\varnothing)\to1.
\]
\end{remark}

\begin{proof}
The proof proceeds in four steps.

\textit{Step 1: Consistency of the plug-in influence functions.}
Fix \(g\in\{1,2\}\). On an event whose probability tends to one, the
AIC-selected model is \(\ell_g^\ast\). On this event, the plug-in and population
influence functions are evaluated under the same copula family.

Write
\[
A_g(z)=C_{g,uu}(z),\qquad
B_g(z)=C_{g,uv}(z),\qquad
D_g(z)=C_{g,u\theta}(z),
\]
and let \(\widehat A_g(z)\), \(\widehat B_g(z)\), and \(\widehat D_g(z)\) be
their plug-in counterparts. By uniform consistency of the empirical marginals,
consistency of \(\widehat\theta_g\), and uniform continuity of the relevant
copula derivatives on the trimmed domain,
\[
\sup_{z\in\Lambda_\varepsilon}
|\widehat A_g(z)-A_g(z)|=o_p(1),
\]
\[
\sup_{z\in\Lambda_\varepsilon}
|\widehat B_g(z)-B_g(z)|=o_p(1),
\]
and
\[
\sup_{z\in\Lambda_\varepsilon}
\|\widehat D_g(z)-D_g(z)\|=o_p(1).
\]
Moreover,
\[
\sup_x|\widehat F_{gX}(x)-F_{gX}(x)|=o_p(1),
\qquad
\sup_y|\widehat F_{gY}(y)-F_{gY}(y)|=o_p(1).
\]

The plug-in estimate \(\widehat\psi_{gi}\) is obtained by replacing the
population score, Hessian, information matrix, and marginal correction terms in
the influence function of the pseudo-likelihood estimator by their sample
counterparts. Under the same smoothness and nonsingularity conditions used to
derive the expansion in Assumption~\ref{ass:copula_first_stage}, these plug-in
components are consistent, and hence
\[
    \frac{1}{n_g}
    \sum_{i=1}^{n_g}
    \|\widehat\psi_{gi}-\psi_g(Z_{gi})\|^2
    =
    o_p(1).
\]
Combining these facts with the boundedness of the copula derivative weights and
\(E\|\psi_g(Z_g)\|^2<\infty\), we obtain
\[
    \sup_{z\in\Lambda_\varepsilon}
    \frac{1}{n_g}
    \sum_{i=1}^{n_g}
    \{\widehat\phi_{g,z}(Z_{gi})-\phi_{g,z}(Z_{gi})\}^2
    =
    o_p(1).
\]
This establishes the required plug-in consistency.

\textit{Step 2: Conditional weak convergence of the multiplier process.}
Define the infeasible multiplier process
\[
\mathbb G_N^\circ(z)
=
\sqrt{1-\widehat\lambda}
\frac{1}{\sqrt{n_1}}
\sum_{i=1}^{n_1}
\xi_{1i}\phi_{1,z}(Z_{1i})
-
\sqrt{\widehat\lambda}
\frac{1}{\sqrt{n_2}}
\sum_{i=1}^{n_2}
\xi_{2i}\phi_{2,z}(Z_{2i}).
\]
The classes \(\{\phi_{g,z}:z\in\Lambda_\varepsilon\}\) are Donsker by the proof
of Theorem~\ref{thm:linear_h}. Therefore, by the conditional multiplier central
limit theorem, jointly across the two independent samples,
\[
    \mathbb G_N^\circ
    \rightsquigarrow_\xi
    \mathbb G_\Delta
    \quad
    \text{in probability in }
    \ell^\infty(\Lambda_\varepsilon).
\]

It remains to replace \(\phi_{g,z}\) by \(\widehat\phi_{g,z}\). The difference
\(\widehat\phi_{g,z}-\phi_{g,z}\) can be decomposed into terms due to
\(\widehat A_g-A_g\), \(\widehat B_g-B_g\), \(\widehat D_g-D_g\), the empirical
centering errors \(\widehat F_{gX}-F_{gX}\) and
\(\widehat F_{gY}-F_{gY}\), and the plug-in error
\(\widehat\psi_{gi}-\psi_g(Z_{gi})\). The first three terms are products of
uniformly \(o_p(1)\) coefficient errors and tight multiplier empirical
processes. The centering terms are products of
\(\sup|\widehat F-F|=o_p(1)\) and \(n_g^{-1/2}\sum_i\xi_{gi}=O_{p,\xi}(1)\).
The last term has conditional variance bounded by
\[
    \frac{1}{n_g}
    \sum_{i=1}^{n_g}
    \|\widehat\psi_{gi}-\psi_g(Z_{gi})\|^2
    =
    o_p(1).
\]
Hence,
\[
    \sup_{z\in\Lambda_\varepsilon}
    |\widehat{\mathbb G}_N^{(m)}(z)-\mathbb G_N^\circ(z)|
    =
    o_{p,\xi}(1)
    \quad
    \text{in probability}.
\]
Thus,
\[
    \widehat{\mathbb G}_N^{(m)}
    \rightsquigarrow_\xi
    \mathbb G_\Delta
    \quad
    \text{in probability in }
    \ell^\infty(\Lambda_\varepsilon).
\]

\textit{Step 3: Consistency of the estimated contact set.}
Under the boundary null, Theorem~\ref{thm:two_sample_sup_limit} gives
\[
    \sup_{z\in\Lambda_\varepsilon}
    |s_N\{\widehat\Delta(z)-\Delta(z)\}|
    =
    O_p(1).
\]
Hence Lemma~\ref{lem:contact_set_consistency} applies directly: for every fixed
\(\eta>0\),
\[
    \Prob\{\Gamma(\Delta)\subseteq\widehat\Gamma_N
    \subseteq\Gamma^\eta(\Delta)\}\to1.
\]

If the null is strict, Remark~\ref{rem:strict_null_contact_set} gives
\(\Prob(\widehat\Gamma_N=\varnothing)\to1\), while
Theorem~\ref{thm:two_sample_sup_limit} gives \(T_N\to_p-\infty\).

\textit{Step 4: Supremum over the estimated contact set and critical values.}
Under the boundary null, Step 3 implies that, with probability approaching one,
\[
\sup_{z\in\Gamma(\Delta)}
\widehat{\mathbb G}_N^{(m)}(z)
\le
\sup_{z\in\widehat\Gamma_N}
\widehat{\mathbb G}_N^{(m)}(z)
\le
\sup_{z\in\Gamma^\eta(\Delta)}
\widehat{\mathbb G}_N^{(m)}(z).
\]
By Step 2, the conditional multiplier process is tight and asymptotically
uniformly equicontinuous on \(\Lambda_\varepsilon\). Since the limiting Gaussian
process has uniformly continuous sample paths on the compact set
\(\Lambda_\varepsilon\), letting \(\eta\downarrow0\) yields
\[
    \sup_{z\in\widehat\Gamma_N}
    \widehat{\mathbb G}_N^{(m)}(z)
    -
    \sup_{z\in\Gamma(\Delta)}
    \widehat{\mathbb G}_N^{(m)}(z)
    =
    o_{p,\xi}(1)
    \quad
    \text{in probability}.
\]
Therefore,
\[
    T_N^{(m)}
    =
    \sup_{z\in\widehat\Gamma_N}
    \widehat{\mathbb G}_N^{(m)}(z)
    \rightsquigarrow_\xi
    \sup_{z\in\Gamma(\Delta)}
    \mathbb G_\Delta(z)
    \quad
    \text{in probability}.
\]

If the distribution of
\(\sup_{z\in\Gamma(\Delta)}\mathbb G_\Delta(z)\) is continuous at
\(c_{1-\alpha}\), the conditional quantile mapping is continuous, and hence
\[
    \widehat c_{1-\alpha}^{\,(m)}
    \overset{p}{\longrightarrow}
    c_{1-\alpha}.
\]

\end{proof}

\subsection*{Proof of Theorem~\ref{thm:bootstrap_cv_consistency}}
To prove Theorem~\ref{thm:bootstrap_cv_consistency}, we first establish two
lemmas. The first shows that bootstrap model averaging is conditionally
first-order equivalent to using the uniquely identified limiting family.

\begin{lemma}[Bootstrap oracle equivalence]
\label{lem:bootstrap_oracle_equivalence}
Suppose Assumptions~\ref{ass:copula_first_stage},
\ref{ass:bootstrap_weights}, and \ref{ass:bootstrap_criterion} hold. For each
population \(g\), define the oracle bootstrap estimator
\[
    \widehat H_g^{*,o}(y\mid x)
    =
    \partial_u C_{\ell_g^\ast}
    \{\widehat F_{gX}^{*}(x),\widehat F_{gY}^{*}(y);
    \widehat\theta_{g\ell_g^\ast}^{*}\}.
\]
Then, conditionally on the observed data and in probability,
\[
    \sup_{(x,y)\in\Lambda_\varepsilon}
    \left|
    \widehat H_g^{*}(y\mid x)
    -
    \widehat H_g^{*,o}(y\mid x)
    \right|
    =
    o_p^{*}(n_g^{-1/2}),
    \qquad g=1,2.
\]
\end{lemma}

\begin{proof}
Fix \(g\in\{1,2\}\), and write
\[
    \widehat Q_{g\ell}
    =
    \widehat M_{g\ell}(\widehat\theta_{g\ell}),
    \qquad
    \widehat Q_{g\ell}^{*}
    =
    \widehat M_{g\ell}^{*}(\widehat\theta_{g\ell}^{*}).
\]
Let
\[
    R_{ng}^{*}
    =
    \max_{1\le\ell\le Q}
    \sup_{\theta\in\Theta_\ell}
    \left|
    \widehat M_{g\ell}^{*}(\theta)
    -
    \widehat M_{g\ell}(\theta)
    \right|.
\]
By the maximizing properties of \(\widehat\theta_{g\ell}\) and
\(\widehat\theta_{g\ell}^{*}\),
\[
    \widehat Q_{g\ell}-R_{ng}^{*}
    \le
    \widehat Q_{g\ell}^{*}
    \le
    \widehat Q_{g\ell}+R_{ng}^{*}.
\]
Consequently,
\[
    \max_{1\le\ell\le Q}
    |\widehat Q_{g\ell}^{*}-\widehat Q_{g\ell}|
    \le R_{ng}^{*}
    =o_p^{*}(1)
\]
conditionally on the data, in probability.

The proof of Lemma~\ref{lem:oracle_equivalence} gives, for every
\(j\ne\ell_g^\ast\),
\[
    \frac{AIC_{gj}-AIC_{g\ell_g^\ast}}{n_g}
    \overset{p}{\longrightarrow}
    2\kappa_{gj},
    \qquad \kappa_{gj}>0.
\]
Because the AIC penalty terms are unchanged in the bootstrap criterion,
\[
\begin{aligned}
&\max_{j\ne\ell_g^\ast}
\left|
\frac{AIC_{gj}^{*}-AIC_{g\ell_g^\ast}^{*}}{n_g}
-
\frac{AIC_{gj}-AIC_{g\ell_g^\ast}}{n_g}
\right| \\
&\qquad\le
4R_{ng}^{*}
=o_p^{*}(1).
\end{aligned}
\]
It follows that, conditionally on the data and in probability,
\[
    \frac{AIC_{gj}^{*}-AIC_{g\ell_g^\ast}^{*}}{n_g}
    \overset{p^{*}}{\longrightarrow}
    2\kappa_{gj},
    \qquad j\ne\ell_g^\ast.
\]

Let \(\kappa_g=\min_{j\ne\ell_g^\ast}\kappa_{gj}>0\). For any fixed
\(c\in(0,\kappa_g)\), with conditional probability approaching one, in
probability,
\[
    AIC_{gj}^{*}-AIC_{g\ell_g^\ast}^{*}
    \ge 2cn_g,
    \qquad j\ne\ell_g^\ast.
\]
Thus \(\ell_g^\ast\) minimizes the bootstrap AIC on this event, and the
bootstrap weights satisfy
\[
    \widehat w_{gj}^{*}
    =O_p^{*}\{\exp(-cn_g)\},
    \qquad
    1-\widehat w_{g\ell_g^\ast}^{*}
    =O_p^{*}\{\exp(-cn_g)\}.
\]

Finally, each copula derivative \(\partial_u C_\ell(u,v;\theta)\) is a
conditional distribution function and therefore takes values in \([0,1]\).
Writing \(\widehat h_{g\ell}^{*}\) for the candidate-specific bootstrap
conditional distribution estimator, we obtain
\[
\begin{aligned}
&\sup_{(x,y)\in\Lambda_\varepsilon}
\left|
\widehat H_g^{*}(y\mid x)-\widehat H_g^{*,o}(y\mid x)
\right| \\
&\quad\le
\sum_{j\ne\ell_g^\ast}\widehat w_{gj}^{*}
=O_p^{*}\{\exp(-cn_g)\}
=o_p^{*}(n_g^{-1/2}).
\end{aligned}
\]
This proves the lemma.
\end{proof}

The second lemma records the conditional expansion of the CMPL estimator in
the limiting family.

\begin{lemma}[Bootstrap expansion of the copula parameter estimator]
\label{lem:bootstrap_parameter_expansion}
Suppose Assumptions~\ref{ass:basic_regular},
\ref{ass:copula_first_stage}, \ref{ass:bootstrap_weights}, and
\ref{ass:bootstrap_criterion} hold. For each
population \(g\in\{1,2\}\), let \(\widehat\theta_g\) denote the CMPL estimator
of the parameter \(\theta_g\) in the limiting selected copula family
\(\ell_g^\ast\). Let \(\widehat\theta_g^{*}\) be the corresponding weighted
bootstrap CMPL estimator based on the bootstrap weights
\(B_g^{*}=(B_{g1}^{*},\ldots,B_{gn_g}^{*})^\top\). Then, conditional on the
observed data,
\[
    \sqrt{n_g}(\widehat\theta_g^{*}-\widehat\theta_g)
    =
    \frac{1}{\sqrt{n_g}}
    \sum_{i=1}^{n_g}
    (B_{gi}^{*}-1)\psi_g(Z_{gi})
    +
    o_p^{*}(1),
    \qquad g=1,2,
\]
where \(\psi_g\) is the influence function in the asymptotic linear expansion
of \(\widehat\theta_g\), and \(o_p^{*}(1)\) denotes convergence to zero in
probability conditional on the data.
\end{lemma}

\begin{proof}
Fix \(g\) and suppress the population subscript where no confusion can arise.
We apply Theorem~1 of \citet{ChengHuang2010Bootstrap} to the semiparametric
criterion
\[
    m_g(\theta,F_{gX},F_{gY})(x,y)
    =
    \log c_{\ell_g^\ast}
    \{F_{gX}(x),F_{gY}(y);\theta\},
\]
whose Euclidean parameter is \(\theta\) and whose nuisance parameter is
\(F_g=(F_{gX},F_{gY})\). Its empirical and weighted empirical versions are
\[
\begin{aligned}
    M_{ng}(\theta)
    =
    \frac{1}{n_g}
    \sum_{i=1}^{n_g}
    \log c_{\ell_g^\ast}
    \{\widehat F_{gX}(X_{gi}),\widehat F_{gY}(Y_{gi});\theta\},
    \\
    M_{ng}^{*}(\theta)
    =
    \frac{1}{n_g}
    \sum_{i=1}^{n_g}
    B_{gi}^{*}
    \log c_{\ell_g^\ast}
    \{\widehat F_{gX}^{*}(X_{gi}),
      \widehat F_{gY}^{*}(Y_{gi});\theta\}.
\end{aligned}
\]
We verify the conditions of that theorem. Assumption~\ref{ass:bootstrap_weights}
is precisely the exchangeability, nonnegativity, normalization, tail, and
empirical-variance requirements W1--W5 of
\citet{ChengHuang2010Bootstrap}, with their scale constant \(c=1\). These
requirements cover the multinomial bootstrap and imply the exchangeably
weighted empirical-process limit of
\citet[Theorem~2.2]{PraestgaardWellner1993Exchangeably}.

The uniform convergence and unique-maximizer conditions in
Assumption~\ref{ass:copula_first_stage} imply
\(\widehat\theta_g\to_p\theta_g\). Assumption~\ref{ass:bootstrap_criterion}
then gives
\(\widehat\theta_g^*\overset{p^*}{\longrightarrow}\theta_g\) in probability.
Because \(\theta_g\) is interior and both estimators maximize their respective
criteria, their parameter-score equations, together with the marginal-score
linearization below, verify the nearly maximizing requirements (7) and (22) of
\citet{ChengHuang2010Bootstrap}.

The differentiability, envelope, and Donsker conditions in
Assumption~\ref{ass:copula_first_stage} give Conditions S1--S2 and
SB1--SB2: they yield the required stochastic equicontinuity, the local
quadratic expansion, the tail control of the local score envelope, and the
bootstrap Donsker property. For the nuisance rates, the ordinary empirical
distribution functions satisfy
\[
    \|\widehat F_{gX}-F_{gX}\|_\infty
    +
    \|\widehat F_{gY}-F_{gY}\|_\infty
    =
    O_p(n_g^{-1/2}),
\]
so Condition S3 holds with \(\gamma=1/2\). Moreover, the weighted empirical
process theorem just cited gives, conditionally on the data and in probability,
\[
    \|\widehat F_{gX}^{*}-\widehat F_{gX}\|_\infty
    +
    \|\widehat F_{gY}^{*}-\widehat F_{gY}\|_\infty
    =
    O_p^{*}(n_g^{-1/2})
    \quad
    \text{in probability}.
\]
The use of \(n_g+1\), rather than \(n_g\), in the pseudo-observations changes
these bounds only by \(O(n_g^{-1})\). The triangle inequality therefore gives
the bootstrap nuisance-rate Condition SB3 with \(\gamma=1/2>1/4\). Finally,
the positive-information Condition I holds because the derivative matrix in
the parameter direction is \(A_g=-J_g\), which is nonsingular, and the
covariance matrix of the adjusted score is positive definite by
Assumption~\ref{ass:copula_first_stage}.

It remains to identify the adjusted score in the cited theorem. Let
\(s_g\) be as in Remark~\ref{rem:psi_g}, and define the population score map
\[
    \mathcal S_g(\theta,F_X,F_Y)
    =
    E\!\left[
    \partial_\theta\log c_{\ell_g^\ast}
    \{F_X(X_g),F_Y(Y_g);\theta\}
    \right].
\]
Its derivative in marginal directions \((h_X,h_Y)\), evaluated at
\((\theta_g,F_{gX},F_{gY})\), is
\[
    \dot{\mathcal S}_{g,F}[h_X,h_Y]
    =
    E\!\left[
    \partial_u s_g(U_g,V_g)h_X(X_g)
    +
    \partial_v s_g(U_g,V_g)h_Y(Y_g)
    \right].
\]
For an observation \(z_0=(x_0,y_0)\), the empirical-CDF influence directions
are
\[
    h_{X,z_0}(x)=\mathbf 1(x_0\le x)-F_{gX}(x),
    \qquad
    h_{Y,z_0}(y)=\mathbf 1(y_0\le y)-F_{gY}(y).
\]
Continuity of the margins and the probability integral transform then give
\[
    \dot{\mathcal S}_{g,F}[h_{X,z_0},h_{Y,z_0}]
    =
    r_{gX}(u_0)+r_{gY}(v_0),
    \qquad
    (u_0,v_0)=\bigl(F_{gX}(x_0),F_{gY}(y_0)\bigr).
\]
Hence the adjusted score \(\widetilde m_{g0}\) in Theorem~1 of
\citet{ChengHuang2010Bootstrap} is, in the present notation,
\[
    \widetilde m_{g0}(Z_g)
    =
    s_g(U_g,V_g)+r_{gX}(U_g)+r_{gY}(V_g),
\]
and \(-A_g^{-1}\widetilde m_{g0}=\psi_g\) by
Remark~\ref{rem:psi_g}. The cited theorem therefore yields
\[
    \sqrt{n_g}(\widehat\theta_g^{*}-\widehat\theta_g)
    =
    \frac{1}{\sqrt{n_g}}
    \sum_{i=1}^{n_g}
    (B_{gi}^{*}-1)\psi_g(Z_{gi})
    +
    o_p^{*}(1).
\]
The argument applies separately to \(g=1\) and \(g=2\); their bootstrap weights
are conditionally independent by Assumption~\ref{ass:bootstrap_weights}.
\end{proof}

\begin{proof}[Proof of Theorem~\ref{thm:bootstrap_cv_consistency}]
The proof proceeds in two steps. First, we establish the conditional weak
convergence of the bootstrap process. Second, we pass from the full process on
\(\Lambda_\varepsilon\) to the supremum over the estimated contact set.

Define the bootstrap process
\[
    \mathbb G_N^{*}(z)
    =
    s_N\{\widehat\Delta^{*}(z)-\widehat\Delta(z)\},
    \qquad z\in\Lambda_\varepsilon .
\]
We first show that, conditional on the data,
\[
    \mathbb G_N^{*}
    \rightsquigarrow_B
    \mathbb G_\Delta
    \quad
    \text{in } \ell^\infty(\Lambda_\varepsilon)
\]
in probability.

By the bootstrap expansion in
Lemma~\ref{lem:bootstrap_parameter_expansion}, the bootstrap oracle equivalence
in Lemma~\ref{lem:bootstrap_oracle_equivalence}, and the original oracle
equivalence in Lemma~\ref{lem:oracle_equivalence}, together with the bootstrap
empirical-process expansion of the marginal distribution estimators, the
bootstrap analogue of Theorem~\ref{thm:linear_h} gives, uniformly over
\(z\in\Lambda_\varepsilon\),
\[
    \sqrt{n_g}\{\widehat H_g^{*}(z)-\widehat H_g(z)\}
    =
    \frac{1}{\sqrt{n_g}}
    \sum_{i=1}^{n_g}
    (B_{gi}^{*}-1)\phi_{g,z}(Z_{gi})
    +
    o_p^{*}(1),
    \qquad g=1,2.
\]
Therefore,
\[
\begin{aligned}
    \mathbb G_N^{*}(z)
    &=
    \sqrt{1-\widehat\lambda}
    \frac{1}{\sqrt{n_1}}
    \sum_{i=1}^{n_1}
    (B_{1i}^{*}-1)\phi_{1,z}(Z_{1i})  \\
    &\quad -
    \sqrt{\widehat\lambda}
    \frac{1}{\sqrt{n_2}}
    \sum_{i=1}^{n_2}
    (B_{2i}^{*}-1)\phi_{2,z}(Z_{2i})
    +
    o_p^{*}(1),
\end{aligned}
\]
uniformly over \(z\in\Lambda_\varepsilon\), where
\(\widehat\lambda=n_1/(n_1+n_2)\).

The classes
\[
    \mathcal F_g=\{\phi_{g,z}:z\in\Lambda_\varepsilon\},
    \qquad g=1,2,
\]
are Donsker by the proof of Theorem~\ref{thm:linear_h}. Since the bootstrap
weights \(B_{gi}^{*}\) satisfy Assumption~\ref{ass:bootstrap_weights}, the
weighted bootstrap empirical process central limit theorem implies that,
conditionally on the data and in probability,
\[
    \frac{1}{\sqrt{n_g}}
    \sum_{i=1}^{n_g}
    (B_{gi}^{*}-1)\phi_{g,\cdot}(Z_{gi})
    \rightsquigarrow_B
    \mathbb G_g
    \quad
    \text{in } \ell^\infty(\Lambda_\varepsilon),
    \qquad g=1,2.
\]
Because the two samples and the bootstrap weights are independent across
populations, and because \(\widehat\lambda\to\lambda\), we obtain
\begin{equation}
\label{eq:bootstrap_process_limit}
    \mathbb G_N^{*}
    \rightsquigarrow_B
    \mathbb G_\Delta
    \quad
    \text{in } \ell^\infty(\Lambda_\varepsilon)
\end{equation}
conditionally on the data, in probability.

We now pass to the supremum over the estimated contact set. Under the boundary
null, \(\sup_{z\in\Lambda_\varepsilon}\Delta(z)=0\). Moreover,
Theorem~\ref{thm:two_sample_sup_limit} implies
\[
    \|s_N(\widehat\Delta-\Delta)\|_{\infty,\Lambda_\varepsilon}=O_p(1).
\]
Hence Lemma~\ref{lem:contact_set_consistency} applies. For every \(\eta>0\),
with probability approaching one,
\[
    \Gamma(\Delta)
    \subseteq
    \widehat\Gamma_N
    \subseteq
    \Gamma^\eta(\Delta).
\]
Therefore, with probability approaching one,
\[
    \sup_{z\in\Gamma(\Delta)}
    \mathbb G_N^{*}(z)
    \le
    \sup_{z\in\widehat\Gamma_N}
    \mathbb G_N^{*}(z)
    \le
    \sup_{z\in\Gamma^\eta(\Delta)}
    \mathbb G_N^{*}(z).
\]
The conditional weak convergence in \eqref{eq:bootstrap_process_limit} implies
conditional tightness and asymptotic uniform equicontinuity of
\(\mathbb G_N^{*}\) on \(\Lambda_\varepsilon\). Since the limiting Gaussian
process has uniformly continuous sample paths on the compact set
\(\Lambda_\varepsilon\), letting \(\eta\downarrow0\) yields
\[
    \sup_{z\in\widehat\Gamma_N}
    \mathbb G_N^{*}(z)
    -
    \sup_{z\in\Gamma(\Delta)}
    \mathbb G_N^{*}(z)
    =
    o_p^{*}(1)
\]
in probability. Hence
\[
    T_N^{*}
    =
    \sup_{z\in\widehat\Gamma_N}
    \mathbb G_N^{*}(z)
    \rightsquigarrow_B
    \sup_{z\in\Gamma(\Delta)}
    \mathbb G_\Delta(z)
\]
conditionally on the data, in probability.

If the distribution of
\(\sup_{z\in\Gamma(\Delta)}\mathbb G_\Delta(z)\) is continuous at its
\((1-\alpha)\)-quantile \(c_{1-\alpha}\), then the conditional quantile
\(\widehat c_{1-\alpha}^{\,*}\) of \(T_N^{*}\) satisfies
\[
    \widehat c_{1-\alpha}^{\,*}
    \overset{p}{\longrightarrow}
    c_{1-\alpha}.
\]
\end{proof}

\end{document}